\documentclass[8pt]{article}
\usepackage{geometry}               
\usepackage{epstopdf}
\usepackage{amsthm}
\usepackage{amsfonts}
\usepackage{MnSymbol}
\usepackage{hyperref}
\usepackage{color}
\usepackage{tikz}
\usepackage{graphicx}
\usepackage[affil-it]{authblk}
\usepackage[applemac]{inputenc} 
\usepackage[toc,page]{appendix}

\newtheorem{theorem}{Theorem}[section]

\newtheorem{proposition}[theorem]{Proposition}
\newtheorem{corollary}[theorem]{Corollary}

\newtheorem{remark}[theorem]{Remark}

\newcommand{\bea}{\begin{eqnarray}}
\newcommand{\eea}{\end{eqnarray}}
\def\beaa{\begin{eqnarray*}}
\def\eeaa{\end{eqnarray*}}
\def\ba{\begin{array}}
\def\ea{\end{array}}
\def\be#1{\begin{equation} \label{#1}}
\def \eeq{\end{equation}}
\def\bsplit{\begin{split}}
\newcommand{\nn}{\nonumber}

\def\les{\lesssim}
\def\ges{\gtrsim}
\def\c{\cdot}

\def\tr{\mbox{tr}}

\newcommand{\nabb}{\nab\mkern-13mu /\,}

\def\ntrap{trap\mkern-18 mu\big/\,}

\def\Db{\dot{\D}}
\def\squared{\dot{\square}}

\renewcommand{\div}{\mbox{div }}

\newcommand{\lapp}{\mbox{$\bigtriangleup  \mkern-13mu / \,$}}

\def\nab{\nabla}

\def\pr{\partial}

\def\dkb{ \, \mathfrak{d}     \mkern-9mu /}
\def\dk{\mathfrak{d}}

\def\a{\alpha}
\def\b{\beta}

\def\de{\delta}

\def\ep{\epsilon}
\def\la{\lambda}
\def\La{\Lambda}

\def\Si{\Sigma}

\def\th{{\theta}}

\def\ka{\kappa}

\def\Up{\Upsilon}

\def\vphi{{\varphi}}

\def\kab{\underline{\kappa}}

\def\AA{{\mathcal A}}
\def\BB{{\mathcal B}}
\def\CC{{\mathcal C}}

\def\EE{{\mathcal E}}

\def\HH{{\mathcal H}}
\def\II{{\mathcal I}}

\def\LL{{\mathcal L}}
\def\MM{{\mathcal M}}

\def\PP{{\mathcal P}}
\def\QQ{{\mathcal Q}}

\def\TT{{\mathcal T}}

\def\D{{\bf D}}

\def\F{{\bf F}}

\def\M{{\bf M}}

\def\Z{{\bf Z}}

\def\g{{\bf g}}

\def\CCb{\underline{\CC}}

\def\qf{\frak{q}}

\def\ff{\frak{f}}

\def\Rc{\check R}

\newcommand{\piX}{\,^{(X)}\pi}

\def\rhoF{\,^{(F)} \hspace{-2.2pt}\rho}

\def\c{\cdot}
\def \f12{\frac 1 2 }

\def\Db{\dot{\D}}
\def\squared{\dot{\square}}

\def\err{\mbox{Err}}

\def\Mor{\mbox{Mor}}
\def\Morr{\mbox{Morr}}

\def\ec{\check{e}}

\def\Fdot{ \dot{F}}

\def\MMdot{\dot{\MM}}

\def\err{\mbox{Err}}

\def\Rbrev{\breve{R}}
        \def\Tbrev{\breve{T}}
        \def\thbrev{\breve{\th}}
        \def\ntrap{trap\mkern-18 mu\big/\,}
        \def\Ed{\dot{E}}

\begin{document}
\title{\large{\textbf{COUPLED GRAVITATIONAL AND ELECTROMAGNETIC PERTURBATIONS OF REISSNER-NORDSTR{\"O}M SPACETIME\\
 IN A POLARIZED SETTING \\
II - Combined estimates for the system of wave equations}}}
\author{\normalsize{Elena Giorgi}}
\date{} 
\maketitle
\begin{abstract}
In a previous paper on coupled gravitational and electromagnetic  perturbations of Reissner-Nordstr{\"o}m spacetime in a polarized setting (\cite{Elena}), we derived a system of wave equations for two independent quantities, one related to the Weyl curvature and one related to the Ricci curvature of the perturbed spacetime. We analyze here the system of coupled wave equations, deriving combined energy-Morawetz and $r^p$-estimates for the system, in the case of small charge.
\end{abstract}

\tableofcontents

\section*{Introduction}
 The final state conjecture in General Relativity for charged black holes states that the Kerr-Newman spacetime is stable under small perturbation of initial data as solutions to the Einstein-Maxwell equation:
\bea \label{Einsteineq}
\operatorname{Ric}(g)_{\mu\nu}=T(F)_{\mu\nu}:=2 F_{\mu \lambda} {F^{\lambda}}_\nu - \frac 1 2 g_{\mu\nu} F^{\a\b} F_{\a\b}
\eea
where $F$ is a $2$-form satisfying Maxwell's equations 
\bea \label{Max}
\nabla_{[\a} F_{\b\gamma]}=0, \qquad \nabla^\a F_{\a\b}=0.
\eea
In the case of charged black holes, gravitational and electromagnetic perturbations have to be considered together, and they are coupled one another. One step towards the proof of the non-linear stability of Kerr-Newman spacetime is the linear stability of the spherically symmetric charged black holes, namely the Reissner-Nordstr{\"o}m solution, corresponding to Kerr-Newman spacetime with zero angular momentum. For a brief review of previous results on stability of the Einstein and the Einstein-Maxwell's equation see the introduction in \cite{Elena}, and references therein.

In the present paper, we continue the project started in \cite{Elena} in the context of stability of Reissner-Nordstr{\"o}m spacetime under polarized perturbations. In \cite{Elena}, we derived a coupled system of wave equations for two independent quantities $\qf$ and $\qf^\F$, related to the Weyl curvature and to the Ricci curvature of the perturbed spacetime respectively. We recall here the main definitions of Section 6.1 of \cite{Elena}. The main quantities, related to the Weyl curvature and to the Ricci curvature of the polarized electrovacuum spacetime, form a hierarchy of five quantities in the following way: 
\bea\label{quantities}
\begin{split}
\psi_0 &= r^2 \kab^2 \a, \\
\psi_1&=\underline{P}(\psi_0), \\
\psi_2&=\underline{P}(\psi_1)=\underline{P}(\underline{P}(\psi_0))=:\qf, \\
\psi_3&= r^2 \kab \ \ff, \\
\psi_4&= \underline{P}(\psi_3)=:\qf^\F
\end{split}
\eea
where the operator $\underline{P}$ is a rescaled null derivative in the ingoing direction $e_3$ defined as 
\bea\label{definition-operator-P}
\underline{P}(f)&=&r \kab^{-1} e_3f +\frac 12 rf.
\eea
The quantities $\psi_2=\qf$ and $\psi_4=\qf^\F$ are of particular importance because they verify Regge-Wheeler type equations which are coupled one another. Notice that $\psi_0$ and $\psi_1$ are lower order terms with respect to $\qf$, in the sense that they can be obtained by $\qf$ integrating along the $e_3$ direction. Similarly, $\psi_3$ is a lower order term with respect to $\qf^\F$. We recall the main theorem in \cite{Elena}, summarizing the system of coupled wave equations.

\begin{theorem}[Theorem 6.12 in \cite{Elena}]\label{theorem-old-Elena} Let $(\M, \g, \Z)$ be an axially symmetric polarized spacetime solution of the Einstein-Maxwell equation, which is a $O(\ep)$-perturbation of Reissner-Nordstr{\"o}m spacetime. Denote its quasi-local charge by $e$. Then there exist $O(\ep^2)$-invariant quantities $\qf$ and $\qf^{\F}$ related to the Weyl curvature and to the Ricci curvature respectively that verify the following coupled system of wave equations, modulo $O(\ep^2)$ terms,
\bea\label{finalsystem}
\begin{cases}
\Big(\square_2+ \ka\kab-10\rhoF^2\Big)\qf= e\Bigg(\frac{2}{r}\lapp_2\qf^{\F}-\frac{2}{r^2}Q\qf^{\F}-\frac{2}{r^2}\ka\kab \underline{P}\qf^{\F} + \frac 1 r \left(5\ka\kab+8\rho+4\rhoF^2\right)\qf^{\F}\Bigg)+e(l.o.t.)_1, \\
\Big(\square_2+\ka\kab+3\rho\Big)\qf^{\F}=e \Bigg(-\frac {2}{ r^3} \qf\Bigg) +e^2 (l.o.t.)_2
\end{cases}
\eea
where $Q$ is a rescaled null derivatives in the outgoing direction defined as
\bea\label{definition-operator-Q}
Q(f)&=&r\kab   e_4f+\frac 1 2r \ka\kab f
\eea 
 and $(l.o.t.)_1$ and $(l.o.t.)_2$ are lower order terms with respect to $\qf$ and $\qf^{\F}$, explicitely,
\beaa
(l.o.t.)_1&=&-6\rho \psi_3+e\left(-\frac{4}{r^3}\psi_1-\frac{2}{r^2} \psi_0 \right)+e^2 \left( -\frac{20}{r^4} \psi_3\right), \\
(l.o.t.)_2&=&\frac{4}{r^3} \psi_3 
\eeaa
\end{theorem}

According to Theorem \ref{theorem-old-Elena}, we can write the system \eqref{finalsystem} in the following concise form:
\bea\label{schematic-system}
\begin{cases}
\Big(\square_2-V_1\Big)\qf= e\ C_1[\qf^\F]+eL_1[\qf^\F]+e^2 L_1[\qf]+N_1[\qf, \qf^\F], \\
\Big(\square_2-V_2\Big)\qf^{\F}=e\  C_2[\qf] +e^2 L_2[\qf^\F]+N_2[\qf, \qf^\F]
\end{cases}
\eea      
where
\bea\label{definition-main-coefficients}
 V_1&=&   - \ka\kab+10\rhoF^2, \\
 V_2&=&-\ka\kab-3\rho, \\
 C_1[\qf^\F]&=&\frac{2}{r}\lapp_2\qf^{\F}-\frac{2}{r^2}Q\qf^{\F}-\frac{2}{r^2}\ka\kab \underline{P}\qf^{\F} + \frac 1 r \left(5\ka\kab+8\rho+4\rhoF^2\right)\qf^{\F}, \\
 C_2[\qf]&=& -\frac {2}{ r^3} \qf, \\
 L_1[\qf]&=& -\frac{4}{r^3}\psi_1-\frac{2}{r^2} \psi_0, \\
 L_1[\qf^\F]&=& -6\rho \psi_3 -e^2 \frac{20}{r^4} \psi_3, \\
 L_2[\qf^\F]&=& \frac{4}{r^3} \psi_3
 \eea
The terms $C$s and $L$s are the linear terms of the equations, all multiplied by the charge of the spacetime, and the terms $N_1$ and $N_2$ are the quadratic, i.e. $O(\ep^2)$, error terms. In particular:
\begin{itemize}
 \item The terms $C_1$ and $C_2$ are the terms representing the coupling between the Weyl curvature and the Ricci curvature. In the wave equation for $\qf$, the coupling term $C_1=C_1[\qf^\F]$ is an expression in terms of $\qf^\F$ and in the wave equation for $\qf^\F$, the coupling term $C_2=C_2[\qf]$ is an expression in terms of $\qf$. 
 \item The terms $L_1$ and $L_2$ collect the lower order terms: in particular $L_1[\qf]$ are lower order terms with respect to $\qf$, i.e. $\psi_0$ and $\psi_1$, while $L_1[\qf^\F]$ and $L_2[\qf^\F]$ are lower terms with respect to $\qf^\F$, i.e. $\psi_3$. The index $1$ or $2$ denotes if they appear in the first or in the second equation.
 \end{itemize}

The aim of this paper is to prove the following theorem.

\begin{theorem}\label{main-theorem} Let $(\M, \g, \Z)$ be an axially symmetric polarized spacetime solution of the Einstein-Maxwell equation, which is a $O(\ep)$-perturbation of Reissner-Nordstr{\"o}m spacetime, and let $\psi_0$, $\psi_1$, $\qf$, $\psi_3$, $\qf^\F$ be its curvature terms as defined in \eqref{quantities}. Then energy boundedness and weighted $r^p$-estimates hold as in Theorem \ref{main-theorem-1}.
\end{theorem}

As remarked in Appendix A of \cite{Elena}, the system \eqref{finalsystem} is also valid for non-polarized electrovacuum perturbations of Reissner-Nordstr{\"o}m spacetime. Therefore, Theorem \ref{main-theorem} is a fundamental step in the proof of linear stability of Reissner-Nordstr{\"o}m spacetime.

The paper is organized as follows. In Section \ref{Section2}, we define the main weighted energies and bulks used in the estimates.
The energy boundedness and weighted $r^p$-estimates for the system are obtained in Section \ref{section-combined-estimates}. The strategy to obtain the combined estimates is the following:
\begin{enumerate}
\item {\bf Separated estimates} Consider the two equations of the system separately as 
\bea\label{first-equation}
\Big(\square_2-V_1\Big)\qf= \M_1[\qf, \qf^\F]
\eea
and 
\bea\label{second-equation}
\Big(\square_2-V_2\Big)\qf^{\F}=\M_2[\qf, \qf^\F]
\eea
where the linear terms of $\M_1[\qf, \qf^\F]$ and $\M_2[\qf, \qf^\F]$ are $O(e)$. More precisely
\beaa
\M_1[\qf, \qf^\F]&=&e\ C_1[\qf^\F]+eL_1[\qf^\F]+e^2 L_1[\qf]+N_1[\qf, \qf^\F], \\
\M_2[\qf, \qf^\F]&=&e\  C_2[\qf] +e^2 L_2[\qf^\F]+N_2[\qf, \qf^\F].
\eeaa
 We apply to both equations separately the standard procedures used to derive energy-Morawetz estimates to equation of the forms $\square_\g \Psi=V\Psi+\M[\Psi]$. Using the Morawetz vector field as multiplier and the $r^p$-method, we derive separated Morawetz estimates and $r^p$-weighted estimates, \eqref{estimate1}-\eqref{estimate1-2}, as long as higher derivative estimates, \eqref{estimate2}-\eqref{estimate2-2}, for the equations \eqref{first-equation} and \eqref{second-equation}. We will summarize the results we need in Section \ref{section-separated-estimates}, omitting the complete proofs for the sake of the clarity of the derivation of the combined estimates. The proofs will appear in \cite{Elena-Future}. 
 
 Notice that these separated estimates will contain on the right hand side terms involving $\M_1$ and $\M_2$ that at this stage are not controlled.  In particular $\M_1$ and $\M_2$ contain both the coupling terms $C$ and the lower order terms $L$. We will not deal with the non-linear terms $N_1$ and $N_2$ in this paper.
\item {\bf Absorption of the coupling terms} The coupling term $C_1[\qf^\F]$ in the first equation involves up to second derivative of $\qf^\F$. On the other hand, the coupling term $C_2[\qf]$ in the second equation contains no derivative of $\qf$. In order to take into account the difference in the presence of derivatives, we consider the $0$th-order estimate for equation \eqref{first-equation} and the $1$st-order estimate for equation \eqref{second-equation} and we add them together.  This operation will create a combined estimate, where the Morawetz bulks on the left hand side of each equations will absorb the coupling term on the right hand side of the other equation, as in Proposition \ref{absorption-coupling-terms-1}. In the trapping region, this absorption is delicate because of the degeneracy of the bulk norms. Neverthless, the special structure of the coupling terms $C_1[\qf^\F]$ and $C_2[\qf]$ implies a cancellation of problematic terms in the trapping region, as in Proposition \ref{absorption-coupling-terms-2}. This is done in Section \ref{section-combined-estimates}.
\item  {\bf Absorption of the lower order terms} To absorb the lower order terms $L_1[\qf]$, $L_1[\qf^\F]$ and $L_2[\qf^\F]$ in the combined estimate, we derive transport estimates for $\psi_0$, $\psi_1$ and $\psi_3$. We make use of the differential relation \eqref{quantities} to get non-degenerate energy estimates in Proposition \ref{transport-estimates}. Using these estimates, we will be able to control the norms involving the lower terms, as done in Proposition \ref{estimates-lot}. 
 This is done in Section \ref{section-lower-order}.
\end{enumerate}
Summing the separated estimates and absorbing the coupling terms and the lower order terms on the right hand side we obtain a combined estimate for the system as in Theorem \ref{main-theorem-1}. 

The presence of the lower order terms in the equations is treated as for the Teukolsky equation in Kerr spacetime, as previously done in \cite{Siyuan} and \cite{TeukolskyDHR}. On the other hand, the absorption of the coupling terms is proper of the electrovacuum case, with coupled and independent gravitational and electromagnetic perturbations.

\bigskip

\textbf{Acknowledgements} The author is grateful to Sergiu Klainerman and Mu-Tao Wang for comments and suggestions, and to Pei-Ken Hung for many helpful discussions.

\addcontentsline{toc}{section}{Introduction}

\section{Preliminaries and definition of main weighted quantities}\label{Section2}
We consider an electrovacuum spacetime $\MM$ which is a $O(\ep)$-perturbation of Reissner-Nordstr{\"o}m spacetime, as defined in Definition 4.1 of \cite{Elena}. Suppose that the original Reissner-Nordstr{\"o}m spacetime is given by 
\beaa
\g_{RN}(m_0, Q_0)&=& -\left(1-\frac{2m_0}{r}+\frac{Q_0^2}{r^2}\right) dt^2+\left(1-\frac{2m_0}{r}+\frac{Q_0^2}{r^2}\right)^{-1}dr^2 +r^2(d\th^2+\sin^2\th d\vphi^2)\\
&=& -\Up(m_0, Q_0) dt^2+\Up(m_0, Q_0)^{-1}dr^2 +r^2(d\th^2+\sin^2\th d\vphi^2)
\eeaa
with initial mass $m_0$ and initial charge $Q_0$, where $\Up(m_0, Q_0)=1-\frac{2m_0}{r}+\frac{Q_0^2}{r^2}$. In particular the horizon and the photon sphere of the unperturbed spacetime are respectively given by
 \beaa
r_{\mathcal{H}}(m_0, Q_0)=m_0+\sqrt{m_0^2-Q_0^2}, \qquad r_P(m_0, Q_0)=\frac{3m_0+\sqrt{9m_0^2-8Q_0^2}}{2}
\eeaa
Defining the modified Hawking mass $\varpi$ and the quasi-local charge $e$ of the perturbed spacetime $\MM$ as in Section 3 of \cite{Elena}, then $|\varpi-m_0| \le \ep_0 $ and  $|e-Q_0| \le \ep_0 $, for a small $\ep_0>0$,  through the entire evolution. For the perturbed spacetime we can define $\Up:=1-\frac{2\varpi}{r}+\frac{e^2}{r^2}$, using the quasi-local mass $\varpi$ and charge $e$.

As in the construction of the perturbed spacetime in \cite{stabilitySchwarzschild}, we assume that the spacetime $\MM$ contains a timelike hypersurface $\TT$ external to the horizon that separates the spacetime in two regions, such that $\MM=\MM_1 \cup \MM_2$, with $\MM_2$ bounded.
The spacetime  $\MM$  is foliated  by  $2$-spheres  $S$  and   comes equipped  with two   scalar functions $r$ and $u$  with $u $ optical   in $\MM_1$ normalized   on   the last incoming  slice  $\CCb_*$, and $r$ being  the area radius   
 of the spheres $S$.
 
We divide  $\MM$  in   the following  regions:
             \begin{enumerate}
             \item The  red shift region   $\MM_{red}=\MM_2$:
             \beaa
             r\in [r_A:=r_{\mathcal{H}}(m_0, Q_0)(1- \de_0), \  r_{\mathcal{H}}(m_0, Q_0)(1+  \de_0)]
             \eeaa
          for a small $\de_0>0$.  
             \item Trapping region  $\MM_{trap} $:
             \beaa
              \frac{5}{6} r_P(m_0, Q_0)\le r \le  \frac{7}{6}r_P(m_0, Q_0)
             \eeaa
             \item  The far region $\MM_{far}$:
             \beaa
               r  \ge R_0
             \eeaa
             with  $R_0$  a fixed  number  $R_0\gg \frac{7}{6}r_P(m_0, Q_0)$.             
             \end{enumerate}
             
   For fixed $R$ we denote by $\MM_{\le R}$ and  $\MM_{\ge R} $ the   regions  defined by $r\le R$ and $r\ge R$. We denote by $\MM_{\ntrap}$ any region outside the trapping region $\MM_{trap}$.

                We foliate   our spacetime   domain $\MM$    by      $\Z$ invariant  hypersurfaces  $\Si(\tau) $ which are:
             \begin{enumerate}
             \item  Incoming null    in $\MM_{red}$,        
              with $e_3$ as null incoming   generator. We denote  this portion 
              $\Si_{red}(\tau)$. 
             \item   Strictly spacelike  in $\MM_{trap} $. We denote  this  portion   by $\Si_{trap}$. 
                          \item     Outgoing null    in  $\MM_{far}=\MM_{r\ge R_0}$
             with $e_4$ as null  outgoing generator.   We denote this portion by $\Si_{far}$ or $\Si_{\ge  R_0}$.
               We can in  fact assume  that $\Si_{r\ge R_0}$          are the level hypersurfaces  of the   optical function $u$.
             \end{enumerate}

We introduce the following vectorfields, used in the derivation of the estimates:
\begin{itemize}
\item $T=\frac 12 (e_4+\Up e_3)$, 
\item $R=\frac 1 2(e_4-\Up e_3)$, 
\item The redshift vectorfield $Y_\HH$ supported in the region $|\Up|\leq 2\de_\HH$.
\item   Let $\th$  a smooth bump  function   equal $1$ on  $  |\Up|\le\de_\HH $  vanishing for  $|\Up|\ge2 \de_\HH$ and define the vectorfields, 
      \bea
       \label{eq:Rc-Tc}
       \begin{split}
  \Rbrev&:=  \th  \frac 1 2 ( e_4-e_3) +(1-\th)        \Up^{-1} R=    \frac 1 2 \left[\thbrev e_4- e_3\right] \\
  \Tbrev&:=\th  \frac 1 2 ( e_4+e_3) +(1-\th)        \Up^{-1} T= \frac 1 2 \left[ \thbrev e_4+ e_3\right]
  \end{split}
      \eea 
      where $\thbrev=\th+\Up^{-1} (1-\th)$.    Note that $2(   |\Rbrev\Psi|^2 +|\Tbrev \Psi|^2)= |e_3\Psi|^2 +\thbrev^2 |e_4\Psi|^2.$
      \item $\ec_4=e_4+\frac 1 r $

\end{itemize}

    We introduce the following  quantities for a symmetric traceless $S$-horizonthal tensor $\Psi$ in  regions    $\MM(\tau_1, \tau_2)\subset \MM$    in  the past of $\Si(\tau_2)$ and in the future of $\Si(\tau_1)$.
    \begin{enumerate}
    \item Energy quantities on $\Si(\tau)$:
    \begin{itemize}
    \item Basic energy quantity
\bea
\label{def:basic-energy}
E[\Psi](\tau)&=& \int_{\Si(\tau)}\bigg( \frac 1 2  (N_\Si, e_3)^2  \,    |e_4 \Psi|^2  +\frac 1 2 (N_\Si, e_4 )^2\,  |e_3\Psi|^2 +|\nabb\Psi|^2 + r^{-2}|\Psi|^2 \bigg)\nn\\
\eea
\item Weighted energy quantities in the far away region
\bea
 \bsplit
  \Ed_{p\,; \,R}[\Psi](\tau):&=  \int_{\Si_{\ge  R}(\tau)}  r^p |\ec_4\Psi|^2,\\
  \Ed'_{p\,; \,R}[\Psi](\tau):&=  \int_{\Si_{\ge  R}(\tau)}  r^p\left(  |\ec_4\Psi|^2+ r^{-2} |\Psi|^2 \right)
  \end{split}
  \eea
  \item Weighted energy quantities
  \bea
  \bsplit
  E_{p}[\Psi](\tau):&=E[\Psi](\tau)+ \Ed_{p\,; \,R}[\Psi](\tau)\\
   E'_{p}[\Psi](\tau):&=E[\Psi](\tau)+ \Ed'_{p\,; \,R}[\Psi](\tau)\
  \end{split}
  \eea

\end{itemize}
    \item  Morawetz bulk quantities in $\MM(\tau_1, \tau_2)$:
    \begin{itemize}
    \item Basic Morawetz bulk 
    \bea
    \label{def:Mor-bulk}
\Mor[\Psi](\tau_1, \tau_2):= \int_{\MM(\tau_1, \tau_2)} \frac{\varpi^2}{r^3} |\Rbrev(\Psi)|^2+ \frac{\varpi}{r^4} |\Psi|^2 + \frac{(r^2-3\varpi r +2e^2)^2}{r^5}\left( |\nabb \Psi|^2+\frac{\varpi^2}{r^2} |\Tbrev\Psi|^2 \right)     
\eea
\item Improved Morawetz bulk
\bea
\bsplit
\Morr[\Psi](\tau_1, \tau_2)&:=&\Mor[\Psi](\tau_1, \tau_2)+\int_{\MM_{far}(\tau_1, \tau_2)} r^{-1-\de}  |e_3(\Psi)|^2\\
&=&\int_{\MM_{trap}} |R(\Psi)|^2+ r^{-2} |\Psi|^2 + \frac{(r^2-3\varpi r +2e^2)^2}{r^4}\left( |\nabb \Psi|^2+\frac{D}{r^2} |T\Psi|^2 \right)\\
 &&\int_{\MM_{\ntrap}} r^{-3} \big(|e_4\Psi|^2+ r^{-1} |\Psi|^2\big) + r^{-1} |\nabb \Psi|^2+r^{-1-\de}  |e_3(\Psi)|^2
 \end{split}
\eea
\item  Weighted bulk norms in the far away region 
\bea
\MMdot_{p\,; \,R}[\Psi](\tau_1, \tau_2):&=\int_{\MM_{\ge  R}(\tau_1, \tau_2) }  r^{p-1}  \left( p | \ec_4(\Psi) |^2 +(2-p)   |\nabb \Psi|^2+   r^{-2}  |\Psi|^2\right)  
\eea
\item Weighted bulk norms
\bea
\MM_{p}[\Psi](\tau_1, \tau_2):&=\Morr[\Psi](\tau_1, \tau_2)+\MMdot_{p\,; \,R}[\Psi](\tau_1, \tau_2)
\eea
\item Non-degenerate bulk norm
\bea
\hat{\MM}_p[\Psi](\tau_1, \tau_2):&=\int_{\MM(\tau_1, \tau_2)} r^{-3+p}|\Psi|^2+r^{-1+p}|e_4(\Psi)|^2+r^{-1+p}|\nabb\Psi|^2
\eea
\end{itemize}
\item Flux   through $\CCb_*(\tau_1, \tau_2)=\CCb_*\cap \MM(\tau_1, \tau_2) $:
\begin{itemize}
\item Basic flux quantity
\bea
  F[\Psi](\tau_1, \tau_2)&=  \int_{\CCb_*(\tau_1, \tau_2)}\big( |e_3\Psi|^2 + |\nabb \Psi|^2 + r^{-2} \Psi^2\big)
  \eea
  \item Weighted  Flux quantity
\bea
\bsplit
 \Fdot_p[\Psi] (\tau_1,\tau_1):&=   \int_{\CCb_*(\tau_1,\tau_2)}  r^p\left(  |\nabb\Psi|^2+ r^{-2} |\Psi|^2\right)\\
 \\
  F_p[\Psi] (\tau_1,\tau_1):&=F[\Psi] (\tau_1,\tau_1)          +    \Fdot_p[\Psi] (\tau_1,\tau_1)
 \end{split}
\eea
  \end{itemize}
  \item Norm for $\M[\Psi]$:  for $p\ge \de$,
\bea\label{definition-norm-M}
\II_{p}[\Psi, \M](\tau_1,\tau_2) :&=& \int_{\MM_{\ntrap}(\tau_1,\tau_2)} r^{1+p}  |\M|^2  +\int _{\MM_{trap}}  ( |R\Psi|  +r^{-1} |\Psi| ) |\M|
\eea
\item Quadratic error term:
\bea
     \err_{\ep,\de}(\tau_1,\tau_2)[\Psi]=O(\ep) \int_{\MM_{\ge R_0}     (\tau_1, \tau_2)}  r^{-1+\de} \left( |e_4\Psi|^2 +|\nabb\Psi|^2 +r^{-2}|\Psi|^2\right)\
\eea
    \end{enumerate}
    
    \begin{remark} Notice that the Morawetz bulk $\MM_{p}[\Psi]$ is degenerate in the trapping region. Notice that the norm for $\M[\Psi]$ given by $\II_{p}[\Psi, \M](\tau_1,\tau_2)$ does not contain the term $\int _{\tau_1}^{\tau_2}  d\tau \int_{\Si_{trap}(\tau)  }  |T\Psi| |\M| $. Indeed, this term has to be kept in the estimates with its sign, in order to be cancelled out in the analysis of the coupling terms.
    \end{remark}
   
   For each of these quantities, we define their higher derivative version. We denote $\dk=\{ e_3, r e_4, \dkb \}$, where $\dkb=r\nabb$ denotes angular derivative.\footnote{Notice the different weight in $r$ between $e_3$ and $e_4$: this is consistent with the asymmetric definitions of the norms.} We define 
   \begin{enumerate}
\item Higher derivative energies:
   \beaa
   E^s[\Psi]&=&\sum_{0\le k\le s} E[\dk^k\Psi], \qquad  \Ed^s_{p\,; \,R}[\Psi]= \sum_{0\le k\le s}\Ed_{p\,; \,R}[\dk^s\Psi], \qquad  \Ed^{'s}_{p\,; \,R}[\Psi]=\sum_{0\le k\le s}\Ed^{'}_{p\,; \,R}[\dk^s\Psi], \\
   E^s_{p}[\Psi]&=& E^s[\Psi]+\Ed^s_{p\,; \,R}[\Psi],  \qquad E^{'s}_{p}[\Psi]=E^s[\Psi]+\Ed^{'s}_{p\,; \,R}[\Psi], 
   \eeaa
   \item Higher derivative Morawetz bulks:
   \beaa
   \Mor^s[\Psi]&=&\sum_{0\le k\le s} \Mor[\dk^k\Psi], \qquad  \Morr^s[\Psi]=\sum_{0\le k\le s} \Morr[\dk^k\Psi], \qquad \MMdot^s_{p\,; \,R}[\Psi]=\sum_{0\le k\le s}\MMdot_{p\,; \,R}[\dk^s\Psi], \\
   \MM^s_{p}[\Psi]&=& \Morr^s[\Psi]+\MMdot^s_{p\,; \,R}[\Psi], \qquad \hat{\MM}^s_p[\Psi]=\sum_{0\le k\le s} \hat{\MM}_p[\dk^k\Psi]
   \eeaa
   \item Higher derivative fluxes:
   \beaa
   F^s[\Psi]&=&\sum_{0\le k\le s} F[\dk^k\Psi], \qquad  \Fdot^s_p[\Psi] = \sum_{0\le k\le s}\Fdot_p[\dk^k\Psi] , \qquad F^s_p[\Psi]=F^s[\Psi]+\Fdot^s_p[\Psi]
   \eeaa
   \item Higher derivative norm for $\MM[\Psi]$:
   \beaa
   \II^s_{p}[\Psi, \M]&=& \II_{p}[\dk^s \Psi, \dk^s\M] 
   \eeaa
   \item Higher derivative error term:
   \beaa
    \err^s_{\ep,\de}(\tau_1,\tau_2)[\Psi]&=& \sum_{0\le k\le s}\err_{\ep,\de}(\tau_1,\tau_2)[\dk^k\Psi] 
   \eeaa
\end{enumerate}
   
These quantities will appear in the energy-Morawetz and $r^p$ weighted estimates of the system \eqref{schematic-system}. 

In the estimates below we will denote $\AA\les \BB$ if there exists a constant $C$ independent of $\ep$ such that $\AA \le C \BB$.

\section{Energy estimates for the system of coupled wave equations}\label{section-combined-estimates}

In this section we prove the energy boundedness, weighted $r^p$-estimates and integrated energy decay for the system \eqref{schematic-system}.

\begin{theorem}[Energy-Morawetz and $r^p$ estimates for the coupled system for $\qf$ and $\qf^\F$]\label{main-theorem-1} Let $\psi_0$, $\psi_1$, $\qf$, $\psi_3$, $\qf^\F$ be curvature terms, as defined in \eqref{quantities}, of  an axially symmetric polarized spacetime which is a $O(\ep)$-perturbation of Reissner-Nordstr{\"o}m spacetime.  Then, there exist constants $C>0$, independent of $\ep$, and a fixed, arbitrarily small,  $\de>0$ with $0<\ep\ll \de$, such that, if $ e \ll \varpi $,  the following estimates hold true in $\MM(\tau_1,\tau_2)$:
\begin{itemize}
\item For $\de\le p\le 2-\de $, we have
       \bea\label{first-estimate-main-theorem}
       \begin{split}
 & \left(E^s_p[\qf] +    E^{s+1}_p[\qf^\F]\right)(\tau_2)       +\left(\MM^s_p[\qf]+\MM^{s+1}_p\right)(\tau_1, \tau_2)+   \left(F^s_p[\qf]+F^{s+1}_p[\qf]\right)(\tau_1, \tau_2)\\
 &\leq  C\Big(\left(E^s_p[\qf] +    E^{s+1}_p[\qf^\F]+E^s_p[\psi_0]+E^s_p[\psi_1]+E^s_p[\psi_2]\right)(\tau_1)   +\err\Big)
 \end{split}
  \eea
  \item For $\de\le p\le 1-\de $, we have 
  \bea\label{second-estimate-main-theorem}
  \begin{split}
   & \left(E'^s_p[\qf] +    E'^{s+1}_p[\qf^\F]\right)(\tau_2)       +\left(\MM^s_p[\qf]+\MM^{s+1}_p\right)(\tau_1, \tau_2)+   \left(F^s_p[\qf]+F^{s+1}_p[\qf]\right)(\tau_1, \tau_2)\\
 &\leq  C\Big(\left(E'^s_p[\qf] +    E'^{s+1}_p[\qf^\F]+E^s_p[\psi_0]+E^s_p[\psi_1]+E^s_p[\psi_2]\right)(\tau_1)   +\err\Big)
 \end{split}
    \eea
\end{itemize}
where $\err$ are quadratic terms given by 
\beaa
\err&=&\err_{\ep, \de}(\tau_1, \tau_2)[\qf]+\err^1_{\ep, \de}(\tau_1, \tau_2)[\qf^\F]+\II_p[\qf, N_1[\qf, \qf^\F]]( \tau_1,\tau_2)+\II^1_p[\qf^\F, N_2[\qf, \qf^\F]]( \tau_1,\tau_2)\\
&&-\int_{\MM_{trap}}\Lambda  T( \qf )\ N_1[\qf, \qf^\F]+\Lambda  T T \qf^\F \c T(N_2[\qf, \qf^\F] )+\Lambda  T R \qf^\F \c R(N_2[\qf, \qf^\F] )+\Lambda  T \nabb \qf^\F \c \nabb(N_2[\qf, \qf^\F] )
\eeaa
\end{theorem}
 
 \begin{proof} According to Theorem \ref{theorem-old-Elena}, the curvature quantities $\psi_0$, $\psi_1$, $\qf$, $\psi_3$, $\qf^\F$ of an $O(\ep)$-perturbation of Reissner-Nordstr{\"o}m spacetime verify the system \eqref{schematic-system}, composed of the two equations \eqref{first-equation} and \eqref{second-equation}.
We will proof estimate \eqref{first-estimate-main-theorem} in the case of $s=0$, and for higher derivatives the proof can be extended, relying on higher order derivative separated estimates. We summarize here the main steps to the proof, assuming the intermediate steps which will be proved in Sections \ref{section-separated-estimates}-\ref{section-lower-order}. We will consider the case of $\de\le p\le 2-\de$. The stronger estimate for $\de\le p\le 1-\de $ is treated in the same way. 
\begin{enumerate}
\item Treating the two equations separately, we derive energy-Morawetz and $r^p$-estimates with standard techniques, as summarized in Theorems \ref{theorem-0th-estimate} and \ref{1st-order-estimates}. The $0$th order estimate \eqref{estimate1} for $\qf$, with $\M_1[\qf, \qf^\F]=e\ C_1[\qf^\F]+eL_1[\qf^\F]+e^2 L_1[\qf]+N_1[\qf, \qf^\F]$, reads 
\bea\label{final-estimate1}
       \begin{split}
       & E_p[\qf](\tau_2)            +\MM_p[\qf](\tau_1,\tau_2)+    F_p[\qf](\tau_1, \tau_2)\\
        &\les E_{p}[\qf](\tau_1)+\left(\II_p[\qf, e  C_1[\qf^\F]]+\II_p[\qf, e   L_1[\qf^\F]]+ \II_p[\qf, e^2  L_1[\qf]]\right)(\tau_1, \tau_2)  \\
        & -\int_{\MM_{trap}}\Lambda  T( \qf )\c \left(e\ C_1[\qf^\F]+eL_1[\qf^\F]+e^2 L_1[\qf] \right) +   \err_0[\qf]
        \end{split}
  \eea
where $\err_0[\qf]:= \err_{\ep, \de}(\tau_1, \tau_2)[\qf]+\II_p[\qf, N_1[\qf, \qf^\F]]( \tau_1,\tau_2)-\int_{\MM_{trap}}\Lambda  T( \qf )\ N_1[\qf, \qf^\F]$ are quadratic terms.

The $1$-st order estimate \eqref{estimate2} for $\qf^\F$, with $\M_2[\qf, \qf^\F]=e\  C_2[\qf] +e^2 L_2[\qf^\F]+N_2[\qf, \qf^\F]$ is
 \bea\label{final-estimate2}
       \begin{split}
      E^1_p[\qf^\F](\tau_2)            +\MM^1_p[\qf^\F](\tau_1,\tau_2)+    F^1_p[\qf^\F](\tau_1, \tau_2)&\les E^1_{p}[\qf^\F](\tau_1)+  \left(\II^1_p[\qf^\F, e  C_2[\qf]] + \II^1_p[\qf^\F, e^2  L_2[\qf^\F]]\right)(\tau_1, \tau_2)   \\
   &-\int_{\MM_{trap}}\Lambda  T T \qf^\F \c T(e\  C_2[\qf] +e^2 L_2[\qf^\F] )\\
   &-\int_{\MM_{trap}}\Lambda  T R \qf^\F \c R(e\  C_2[\qf] +e^2 L_2[\qf^\F])\\
   &-\int_{\MM_{trap}}\Lambda  T \dkb \qf^\F \c \dkb(e\  C_2[\qf] +e^2 L_2[\qf^\F] )+\err_1[\qf^\F]
   \end{split}
  \eea
where $\err_1[\qf^\F]:= \err^1_{\ep, \de}(\tau_1, \tau_2)[\qf^\F]+\II^1_p[\qf^\F, N_2[\qf, \qf^\F]]( \tau_1,\tau_2)-\int_{\MM_{trap}}\Lambda  T T \qf^\F \c T(N_2[\qf, \qf^\F] )+\Lambda  T R \qf^\F \c R(N_2[\qf, \qf^\F] )+\Lambda  T \nabb \qf^\F \c \nabb(N_2[\qf, \qf^\F] )$ are quadratic terms.

This is done in Section \ref{section-separated-estimates}.

We sum a multiple of the first estimate to the second estimate. In particular, taking $2 \times \eqref{final-estimate1}+ \eqref{final-estimate2}$, and keeping the constant explicit for the integrals in the trapping region on the right hand side we obtain
\bea
       \begin{split}\label{summed-estimates}
      &  \left(E_p[\qf]+E^1_p[\qf^\F]\right)(\tau_2)            +\left(\MM_p[\qf]+ \MM^1_p[\qf^\F]\right)(\tau_1,\tau_2)+   \left( F_p[\qf] +    F^1_p[\qf^\F]\right)(\tau_1, \tau_2)\\
      &\les \left(E_p[\qf]+E^1_p[\qf^\F]\right)(\tau_1)+\err_0[\qf] +\err_1[\qf^\F]\\
      &+\left(\II_p[\qf, e  C_1[\qf^\F]]+\II^1_p[\qf^\F, e  C_2[\qf]]\right)(\tau_1, \tau_2)-\int_{\MM_{trap}}\Lambda  T R \qf^\F \c R(e\  C_2[\qf] )  \\
      &-e \Lambda \int_{\MM_{trap}}2  T( \qf )\c  C_1[\qf^\F]+ T T \qf^\F \c T( C_2[\qf])+  T \dkb \qf^\F \c \dkb( C_2[\qf]) \\
        &+\left(\II_p[\qf, e   L_1[\qf^\F]]+ \II_p[\qf, e^2  L_1[\qf]]+\II^1_p[\qf^\F, e^2  L_2[\qf^\F]]\right)(\tau_1, \tau_2)  \\
                 & -e \Lambda \int_{\MM_{trap}}2 T( \qf )\c \left(L_1[\qf^\F]+e L_1[\qf] \right) - e T T \qf^\F \c T( L_2[\qf^\F] )-  e T R \qf^\F \c R( L_2[\qf^\F])- e  T \nabb \qf^\F \c \nabb( L_2[\qf^\F] )
        \end{split}
  \eea
The third and fourth line of estimate \eqref{summed-estimates} concern the coupling terms, while the last two lines contain the lower order terms, and they are all $O(e)$. We will be able to absorb completely the coupling terms by the left hand side of the estimate for small charge and using a cancellation, as shown in Step 2. We will absorb the lower order terms modulo a necessary bound on initial energy for the lower order quantities $\psi_0$, $\psi_1$, $\psi_3$, as done in Step 3.

\item Because of the structure of the coupling terms, we are able to absorb most of those terms appearing on the right hand side of estimate \eqref{summed-estimates} by the Morawetz bulks and the energies of $\qf$ and $\qf^\F$ on the left hand side. In particular, outside the trapping region, this absorption is straightforward for all terms. In the trapping region though, there are problems due to the degeneracy of the Morawetz bulks, in $T$ and $\nabb$. Specifically, the terms $\int_{\MM_{trap}}2  T( \qf )\c  C_1[\qf^\F]- T T \qf^\F \c T( C_2[\qf])-  T \nabb \qf^\F \c \nabb( C_2[\qf]) $ cannot be aabsorbed by the left hand side. Neverthless, due to the particular structure of the higher order coupling terms,  which is related to the fact that the system composed with higher order terms is diagonalizable, these terms get cancelled out (see Remark \ref{diagonalizable}). In Proposition \ref{absorption-coupling-terms-1} we obtain 
\bea\label{coupling-terms-estimate1}
\begin{split}
&\left(\II_p[\qf, e  C_1[\qf^\F]]+\II^1_p[\qf^\F, e  C_2[\qf]]\right)(\tau_1, \tau_2)-e \Lambda \int_{\MM_{trap}}  T R \qf^\F \c R( C_2[\qf] ) \\
&\les e \sup_{\tau\in[\tau_1, \tau_2]} \left( E[\qf](\tau)+ E^1[\qf^\F](\tau)\right)+e \left(\MM_p[\qf]+\MM^1_p[\qf^\F]\right)(\tau_1, \tau_2)
\end{split}
\eea 
and in Proposition \ref{absorption-coupling-terms-2} we obtain 
\bea\label{coupling-terms-estimate2}
\begin{split}
&-e \Lambda \int_{\MM_{trap}}2  T( \qf )\c  C_1[\qf^\F]+ T T \qf^\F \c T( C_2[\qf])+  T \dkb \qf^\F \c \dkb( C_2[\qf]) \\
&\les  e \sup_{\tau\in[\tau_1, \tau_2]} \left( E[\qf](\tau)+ E^1[\qf^\F](\tau)\right)+e\left(\MM_p[\qf]+\MM^1_p[\qf^\F] + \hat{\MM}[\psi_3]\right)(\tau_1, \tau_2)
\end{split}
\eea

This is done in Section \ref{section-coupling}. 

\item In order to control the lower order terms of the equations, we use the differential relations existing between the quantities $\psi_0, \psi_1, \qf, \psi_3, \qf^\F$. We derive in Proposition \ref{transport-estimates} transport estimates for $\psi_0, \psi_1, \psi_3$ and in Proposition \ref{estimates-lot} we estimate the norms involving the lower order terms. We arrive to the following estimate in Corollary \ref{corollary-lot}: 
\bea\label{absorb-lot}
\begin{split}
&\left(\II_p[\qf, e   L_1[\qf^\F]]+ \II_p[\qf, e  L_1[\qf]]+\II^1_p[\qf^\F, e^2  L_2[\qf^\F]]\right)(\tau_1, \tau_2)  \\
                 & -e \Lambda \int_{\MM_{trap}}2 T( \qf )\c \left(L_1[\qf^\F]+e L_1[\qf] \right) - e T T \qf^\F \c T( L_2[\qf^\F] )-  e T R \qf^\F \c R( L_2[\qf^\F])- e  T \nabb \qf^\F \c \nabb( L_2[\qf^\F] )\\
&\les e\sup_{\tau\in[\tau_1, \tau_2]} \left(E[\qf](\tau)+E^1[\qf^\F](\tau)\right)+e \MM^1_p[\qf^\F](\tau_1, \tau_2)+ \left(E_p[\psi_0]+E_p[\psi_1]+E_p[\psi_3]\right)(\tau_1)
                 \end{split}
\eea
This is done in section \ref{section-lower-order}.
\end{enumerate}

Back to estimate \eqref{summed-estimates}, using \eqref{coupling-terms-estimate1} and \eqref{coupling-terms-estimate2} for the coupling terms and \eqref{absorb-lot} for the lower order terms, we obtain
\bea
       \begin{split}
      &  \left(E_p[\qf]+E^1_p[\qf^\F]\right)(\tau_2)            +\left(\MM_p[\qf]+ \MM^1_p[\qf^\F]\right)(\tau_1,\tau_2)+   \left( F_p[\qf] +    F^1_p[\qf^\F]\right)(\tau_1, \tau_2)\\
      &\les \left(E_p[\qf]+E^1_p[\qf^\F]+E_p[\psi_0]+E_p[\psi_1]+E_p[\psi_3]\right)(\tau_1)+\err\\
      &+e \sup_{\tau\in[\tau_1, \tau_2]} \left( E[\qf](\tau)+ E^1[\qf^\F](\tau)\right)+e\left(\MM_p[\qf]+\MM^1_p[\qf^\F] \right)(\tau_1, \tau_2) 
        \end{split}
  \eea
where the $\hat{\MM}_p[\psi_3]$ appearing in the coupling term was absorbed using the transport estimates. For $e \ll \varpi$, we can absorb the energies and the bulks of $\qf$ and $\qf^\F$ on the left hand side,
which proves \eqref{first-estimate-main-theorem}. The proof of \eqref{second-estimate-main-theorem} is identical.
\end{proof}

In what follows, we prove the estimates in the intermediate steps of the proof of Theorem \ref{main-theorem-1}.

\subsection{Energy estimates for the two equations separately}\label{section-separated-estimates}
The derivation of the Morawetz and $r^p$-weighted estimates for equations \eqref{first-equation} and \eqref{second-equation} is based on the following Proposition. 
\begin{proposition}
\label{prop:qf-tensorial'}
Consider  symmetric  traceless          S-horizonthal   tensor
 $\Psi_{AB}$   and $\M[\Psi]_{AB}$ verifying the spacetime equation
\bea
\label{eq:tensor-waveqf}
\squared_\g \Psi=V\Psi +\M[\Psi]
\eea
Define  the energy momentum tensor associated to the equation \eqref{eq:tensor-waveqf} as
\bea\label{defenergymomentum}
 \QQ_{\mu\nu}:&=&\Db_\mu  \Psi \c \Db _\nu \Psi 
          -\frac 12 \g_{\mu\nu} \left(\Db_\la \Psi\c\Db^\la \Psi + V\Psi \c \Psi\right)\\
             &=&\Db_\mu  \Psi \c \Db _\nu \Psi    -\frac 12 \g_{\mu\nu}  \LL(\Psi)
          \eea
       Let $X= a(r) e_3+b(r) e_4$ be a vectorfield, $w$ a scalar  function and $M$  a one form. Defining
 \beaa
 \PP_\mu[X, w, M]&=&\QQ_{\mu\nu} X^\nu +\frac 1 2  w \Psi \Db_\mu \Psi -\frac 1 4\Psi^2   \pr_\mu w +\frac 1 4 \Psi^2 M_\mu,
  \eeaa
 then, in an axially symmetric polarized spacetime, 
  \bea
  \label{le:divergPP-gen}
  \begin{split}
  \D^\mu  \PP_\mu[X, w, M] &= \frac 1 2 \QQ  \c\piX - \frac 1 2 X( V )  \Psi\c \Psi+\frac 12  w \LL[\Psi] -\frac 1 4\Psi^2   \square_g  w\\
   &+\frac 1 4  \Db^\mu (\Psi^2 M_\mu)      +  \left(X( \Psi )+\frac 1 2   w \Psi\right)\c \M[\Psi] 
   \end{split}
 \eea
\begin{proof} Standard computations.
\end{proof}

\end{proposition}

We summarize in this section the energy-Morawetz and the $r^p$-weighted estimates for the two equations separately. For both equations of the system \eqref{schematic-system}, we can derive higher derivative energy-Morawetz and higher weighted $r^p$ estimates. Nevertheless, because of the structure and dependence of derivatives of $\M_1[\qf, \qf^\F]$ and $\M_2[\qf, \qf^\F]$, we will need to combine the $0$th order estimate for the first equation \eqref{first-equation} with the $1$st order estimate for the second equation \eqref{second-equation}. In general we will need to combine the $s$-th order estimate for the first with the $(s+1)$-th order estimate for the second. 

To obtain the cancellation that will be needed in the absorption of the coupling terms, we keep a part of the term $\left(X( \Psi )+\frac 1 2   w \Psi\right)\c \M[\Psi] $ as it is in the estimates, with its sign, and we bound the rest of it with the norm $\II_p[\Psi, \M]$.

\begin{theorem}[$0$th order energy-Morawetz and $r^p$ weighted estimates for $\qf$]\label{theorem-0th-estimate} Consider the equation $$\square_\g\qf=V_1\qf+ \M_1[\qf, \qf^\F]$$ in a $O(\ep)$-perturbation of Reissner-Nordstr{\"o}m spacetime, where $V_1=   - \ka\kab+10\rhoF^2$. Then there exists a fixed, arbitrarily small,  $\de>0$ with $0<\ep\ll \de$, and $\Lambda>\de^{-1}$ such that the following estimates hold true in $\MM(\tau_1,\tau_2)$:
\begin{itemize}
\item For $\de\le p\le 2-\de $, we have
       \bea\label{estimate1}
       \begin{split}
        E_p[\qf](\tau_2)            +\MM_p[\qf](\tau_1,\tau_2)+    F_p[\qf](\tau_1, \tau_2)&\les E_{p}[\qf](\tau_1)+  \II_p[\qf, \M_1[\qf, \qf^\F]] +   \err_{\ep, \de}(\tau_1, \tau_2)[\qf] \\
        & -\int_{\MM_{trap}}\Lambda  T( \qf )\c \M_1[\qf, \qf^\F] 
        \end{split}
  \eea
  \item For $\de\le p\le 1-\de $, we have 
  \bea\label{estimate1-2}
  \begin{split}
   E'_p[\qf](\tau_2)            +\MM_p[\qf](\tau_1,\tau_2)+    F_p[\qf](\tau_1, \tau_2)&\les E'_{p}[\qf](\tau_1)+  \II_p[\qf, \M_1[\qf, \qf^\F]] +   \err_{\ep, \de}(\tau_1, \tau_2)[\qf] \\
        & -\int_{\MM_{trap}}\Lambda  T( \qf )\c \M_1[\qf, \qf^\F]   
        \end{split}
    \eea
\end{itemize}
\end{theorem}
\begin{proof} We sketch here the proof, to appear in \cite{Elena-Future}. We apply Proposition \ref{prop:qf-tensorial'} to equation $\square_\g \qf=V_1\qf +\M_1[\qf, \qf^\F]$ with triplet
\beaa
   (X, w, M):=  (f _\de R,  w_\de, 2 h R  )  +\ep_\HH( Y_\HH, 0, 0)  + (0, 0, 2 h \Rc)  + (\Lambda T, 0, 0)+(f_p e_4,  \frac{2f_p}{r}, \frac{2 f_p'}{r} e_4)
  \eeaa
   for well-chosen functions $f_\de$, $w_\de$, $h$, $f_p=\theta(r) r^p$ and well chosen constants $\ep_\HH$, $\Lambda$. Defining 
     \beaa
 \EE[X, w, M](\qf)&:=& \D^\mu  \PP_\mu[X, w, M] -   \left(X( \qf )+\frac 1 2   w \qf\right)\c \M_1[\qf, \qf^\F]   
  \eeaa
  applying divergence theorem we obtain
    \bea
    \label{eq:boundariesMM}
    \bsplit
             & \int_{r= r_A} \PP \c N_A + \int_{\Si_2}\PP\c N_\Si +\int_{\MM(\tau_1, \tau_2)} \EE        +\int_{\CCb_*} \PP \c   e_3  = \int_{\Si_1}\PP\c N_\Si- \int_{\MM(\tau_1, \tau_2)}  \left(X( \qf )+\frac 1 2   w \qf\right)\c \M_1[\qf, \qf^\F]  
              \end{split}
              \eea
              where      $\EE=  \EE[X, w, M](\qf)$  and             $N_A$  is the unit normal to the spacelike hypersurface  $r=r_A$. 
            We can show that,  for  $C\gg \de^{-1}$ independent of $\ep$,
              \bea
       \label{equation:LBforEE}
       \begin{split}
        \int_{\MM(\tau_1, \tau_2)}  \EE&\ge C^{-1}  \MM_p[\qf]( \tau_1, \tau_2)+O\big(  \ep  \La      \big)   \err_{\ep, \de}[\qf](\tau_1,\tau_2), \\
        \int_{r=r_A}  \PP\c   N_A &\ge 0, \\
         \int_{\Si(\tau)}   \PP\c N_\Si &\ges    E_p[\Psi](\tau), \\
          \int_{\CCb_*(\tau_1, \tau_2)}   \PP\c e_3  &\ge    \frac{\La}{4}F[\Psi](\tau_1, \tau_2)
        \end{split}
       \eea
We analyze the inhomogeneous term $ - \int_{\MM(\tau_1, \tau_2)}  \left(X( \qf )+\frac 1 2   w \qf\right)\c \M_1[\qf, \qf^\F]$. Since the vectorfield $X$ is given by    $X=  f_\de   R +\ep_\HH Y_\HH +\La T+f_p e_4$, and $f_p$ is supported in $r\ge R$, we separate the term involving $T\qf$ and bound the other ones with their absolute value:
\bea\label{inhomogeneous-term}
\begin{split}
- \int_{\MM(\tau_1, \tau_2)}  \left(X( \qf )+\frac 1 2   w \qf\right)\c \M_1[\qf, \qf^\F]&\le - \int_{\MM(\tau_1, \tau_2)}  \Lambda T(\qf )\c \M_1[\qf, \qf^\F]+C\int_{\MM(\tau_1,\tau_2)}\left(  |\Rbrev \qf| + r^{-1} |\qf|\right)|\M_1[\qf, \qf^\F]|\\
&+C\int_{\MM_{\ge R}(\tau_1,\tau_2)  } r^p |e_4\qf||\M_1[\qf, \qf^\F]|
\end{split}
\eea
We separate the first and second integral on the right of \eqref{inhomogeneous-term} in the integral in the trapping region and outside the trapping region. Outside the trapping region, the integral $- \int_{\MM(\tau_1, \tau_2)}  \Lambda T(\qf )\c \M_1[\qf, \qf^\F]$ can also be bounded by its absolute value:
\bea
\begin{split}
- \int_{\MM(\tau_1, \tau_2)}  \left(X( \qf )+\frac 1 2   w \qf\right)\c \M_1[\qf, \qf^\F]&\les - \int_{\MM_{trap}}  \Lambda T(\qf )\c \M_1[\qf, \qf^\F]+\int_{\MM_{trap}}\left(  |\Rbrev \qf| + r^{-1} |\qf|\right)|\M_1[\qf, \qf^\F]|\\
&+\int_{\MM_{\ntrap}}\left(  |\Rbrev \qf| +|\Tbrev \qf|+ r^{-1} |\qf|\right)|\M_1[\qf, \qf^\F]|\\
&+\int_{\MM_{\ge R}(\tau_1,\tau_2)  } r^p |e_4\qf||\M_1[\qf, \qf^\F]|
\end{split}
\eea
The integrals outside the trapping region the integral can be separated using Cauchy-Schwarz:
         \beaa
\int_{\MM_{\ntrap}(\tau_1,\tau_2)}  ( |\Rbrev\qf| +  |\Tbrev\qf|  +r^{-1} |\qf| ) |\M_1[\qf, \qf^\F]| &\le & \la \int_{\MM_{\ntrap}}  r^{-1-\de}  ( |\Rbrev\qf|^2 + |\Tbrev\qf|^2    +r^{-2} |\qf|^2 )+\la^{-1}  \int_{\MM_{\ntrap}}  r^{1+\de} |\M_1[\qf, \qf^\F]|^2 , \\
\int_{\MM_{\ge R}(\tau_1,\tau_2)  } r^p |e_4\qf||\M_1[\qf, \qf^\F]|&\le&  \lambda \int_{\MM_{\ge R}(\tau_1,\tau_2)  } r^{p-1} |e_4 \qf|^2+\lambda^{-1} \int_{\MM_{\ge R}(\tau_1,\tau_2)  } r^{p+1} |\M_1[\qf, \qf^\F]|^2
         \eeaa
    and for $\lambda$ small enough the first  integrals on the right can be absorbed in the Morawetz bulk $\MM_p[\qf]$.
    We thus arrive to  the estimate,
         \beaa
              E_p[\qf](\tau_2)            +\MM_p[\qf](\tau_1,\tau_2)+F_p[\qf](\tau_1,\tau_2)&\les&     E_p[\qf](\tau_2)  +\err_\ep(\tau_1, \tau_2)- \int_{\MM_{trap}}  \Lambda T(\qf )\c \M_1[\qf, \qf^\F]\\
              &&+\int _{\MM_{trap}}  ( |\Rbrev\qf|   +r^{-1} |\qf| ) |\M_1[\qf, \qf^\F]|+\int_{\MM_{\ntrap}}  r^{1+p} |\M_1[\qf, \qf^\F]|^2
              \eeaa
             and  recalling the definition of $\II_p[\qf,\M_1[\qf, \qf^\F]]$ in \eqref{definition-norm-M}, we obtain the desired estimate. 
 
             \end{proof}

With the same techniques, we derive the estimates for $\qf^\F$, independently of the ones for $\qf$. 

\begin{theorem}[$1$st order energy-Morawetz and $r^p$ weighted estimates for $\qf^\F$]\label{1st-order-estimates} Consider the equation $$\square_\g\qf^\F=V_2\qf^\F+ \M_2[\qf, \qf^\F]$$ in a $O(\ep)$-perturbation of Reissner-Nordstr{\"o}m spacetime, where $V_2=   - \ka\kab-3\rho$. Then there exists a fixed, arbitrarily small,  $\de>0$ with $0<\ep\ll \de$, such that the following estimates hold true in $\MM(\tau_1,\tau_2)$:
\begin{itemize}
\item For $\de\le p\le 2-\de $, we have
       \bea\label{estimate2}
       \begin{split}
      E^1_p[\qf^\F](\tau_2)            +\MM^1_p[\qf^\F](\tau_1,\tau_2)+    F^1_p[\qf^\F](\tau_1, \tau_2)&\les E^1_{p}[\qf^\F](\tau_1)+  \II^1_p[\qf^\F,\M_2[\qf, \qf^\F]]( \tau_1,\tau_2) +   \err^1_{\ep, \de}(\tau_1, \tau_2)[\qf^\F]  \\
   &-\int_{\MM_{trap}}\Lambda  T T \qf^\F \c T(\M_2[\qf, \qf^\F] )\\
   &-\int_{\MM_{trap}}\Lambda  T R \qf^\F \c R(\M_2[\qf, \qf^\F] )\\
   &-\int_{\MM_{trap}}\Lambda  T \dkb \qf^\F \c \dkb(\M_2[\qf, \qf^\F] )
   \end{split}
  \eea
  \item For $\de\le p\le 1-\de $, we have 
  \bea\label{estimate2-2}
  \begin{split}
   E^{'1}_p[\qf^\F](\tau_2)            +\MM^1_p[\qf^\F](\tau_1,\tau_2)+  F^1_p[\qf^\F](\tau_1, \tau_2)  &\les E^{'1}_{p}[\qf^\F](\tau_1)+  \II^1_p[\qf^\F, \M_2[\qf, \qf^\F]]( \tau_1,\tau_2) + \err^1_{\ep, \de}(\tau_1, \tau_2)[\qf^\F]  \\
    &-\int_{\MM_{trap}}\Lambda  T T \qf^\F \c T(\M_2[\qf, \qf^\F] )\\
   &-\int_{\MM_{trap}}\Lambda  T R \qf^\F \c R(\M_2[\qf, \qf^\F] )\\
   &-\int_{\MM_{trap}}\Lambda  T \dkb \qf^\F \c \dkb(\M_2[\qf, \qf^\F] )
   \end{split}
    \eea
\end{itemize}
\end{theorem}
\begin{proof} The proof proceeds as in Theorem \ref{theorem-0th-estimate} after commuting the equation $\square_\g\qf^\F=V_2\qf^\F+ \M_2[\qf, \qf^\F]$ with $T$, $R$ and $\dkb$. To appear in \cite{Elena-Future}.
\end{proof}

Theorems  \ref{theorem-0th-estimate} and \ref{1st-order-estimates} have technical but standard proofs which will be postponed to appear in \cite{Elena-Future} for the sake of the clarity of the combined estimates of the system.

\subsection{Estimates for the coupling terms}\label{section-coupling}
The aim of this section is to prove \eqref{coupling-terms-estimate1} and \eqref{coupling-terms-estimate2} in Step 2 of the proof of Theorem \ref{main-theorem}.

By \eqref{definition-operator-P}, \eqref{definition-operator-Q} and \eqref{definition-main-coefficients}, the coupling terms are given by  
\beaa
 C_1[\qf^\F]&=&\frac{2}{r}\lapp_2\qf^{\F}-\frac{2}{r}\kab   e_4\qf^{\F}-\frac{2}{r}\ka e_3\qf^{\F} + \frac 1 r \left(3\ka\kab+8\rho+4\rhoF^2\right)\qf^{\F}, \\
 C_2[\qf]&=& -\frac {2}{ r^3} \qf
\eeaa
Recall that, in an $O(\ep)$ perturbation of Reissner-Nordstr{\"o}m spacetime, $\ka-\overline{\ka}, \kab-\overline{\kab}, \rho-\overline{\rho}, \rhoF-\overline{\rhoF}=O(\ep)$, therefore we can write schematically
\beaa
\ka&=& \frac 1 r +O(\ep), \qquad \kab=\frac 1 r +O(\ep), \qquad \rho=\frac{\varpi}{r^3}+O(\ep), \qquad \rhoF=\frac{e}{r^2}+O(\ep)
\eeaa
Outside the trapping region we will schematize\footnote{Ignoring the quadratic terms} the coupling terms as
\bea\label{schematic-coupling-terms}
C_1[\qf^\F]&=&\Big[\frac{1}{r}\lapp_2\qf^{\F}, \frac{1}{r^2}   e_4\qf^{\F},\frac{1}{r^2} e_3\qf^{\F}, \frac{1}{r^3}\qf^{\F}\Big], \\
 C_2[\qf]&=& \Big[\frac {1}{ r^3} \qf\Big]
\eea

The coupling terms appear in the equations with a good decay in $r$, therefore outside the trapping region they can be easily absorbed by the Morawetz bulks of $\qf$ and $\qf^\F$. In the trapping region though the Morawetz bulks are degenerate, therefore they can't absorb all second derivatives of $\qf^\F$, in particular $\lapp \qf^\F$. We need to use the equations to absorb the terms we need. It turns out that the structure of the coupling terms allows a very convenient cancellation of the bad terms.

We first prove the absorption of the non-degenerate terms, as in \eqref{coupling-terms-estimate1}.

\begin{proposition}\label{absorption-coupling-terms-1}With the notations above, for all $\de \le p \le 2-\de$, we have
\beaa
&&\left(\II_p[\qf, e  C_1[\qf^\F]]+\II^1_p[\qf^\F, e  C_2[\qf]]\right)(\tau_1, \tau_2)-e \Lambda \int_{\MM_{trap}}  T R \qf^\F \c R( C_2[\qf] ) \\
&\les& e \sup_{\tau\in[\tau_1, \tau_2]} \left(E[\qf](\tau) +E^1[\qf^\F]\right)+e \left(\MM_p[\qf]+\MM^1_p[\qf^\F]\right)(\tau_1, \tau_2)
\eeaa
\end{proposition}
\begin{proof} Recall that, by definition, the Morawetz bulks $\MM_p[\qf](\tau_1,\tau_2)$ and $\MM^1_p[\qf^\F](\tau_1,\tau_2)$ are given by 
\beaa
\MM_p[\qf](\tau_1,\tau_2)&=& \int_{\MM_{trap}} |R(\qf)|^2+ r^{-2} |\qf|^2 + \frac{(r^2-3\varpi r +2e^2)^2}{r^4}\left( |\nabb \qf|^2+\frac{D}{r^2} |T\qf|^2 \right)\\
 &&+\int_{\MM_{\ntrap}} r^{-3} \big(|e_4\qf|^2+ r^{-1} |\qf|^2\big) + r^{-1} |\nabb \qf|^2+r^{-1-\de}  |e_3(\qf)|^2 \\
 &&+\int_{\MM_{\ge  R}(\tau_1, \tau_2) }  r^{p-1}  \left( p | \ec_4(\qf) |^2 +(2-p)   |\nabb \qf|^2+   r^{-2}  |\qf|^2\right)
 \eeaa
 and
 \beaa
\MM^1_p[\qf^\F](\tau_1,\tau_2)&=& \int_{\MM_{trap}} |R(\qf^\F)|^2+ r^{-2} |\qf^\F|^2 + \frac{(r^2-3\varpi r +2e^2)^2}{r^4}\left( |\nabb \qf^\F|^2+\frac{D}{r^2} |T\qf^\F|^2 \right)\\
 &&+\int_{\MM_{\ntrap}} r^{-3} \big(|e_4\qf^\F|^2+ r^{-1} |\qf^\F|^2\big) + r^{-1} |\nabb \qf^\F|^2+r^{-1-\de}  |e_3(\qf^\F)|^2 \\
 &&+\int_{\MM_{\ge  R}(\tau_1, \tau_2) }  r^{p-1}  \left( p | \ec_4(\qf^\F) |^2 +(2-p)   |\nabb \qf^\F|^2+   r^{-2}  |\qf^\F|^2\right)\\
 && +\int_{\MM_{trap}} |R(\dk\qf^\F)|^2+ r^{-2} |\dk\qf^\F|^2 + \frac{(r^2-3\varpi r +2e^2)^2}{r^4}\left( |\nabb \dk\qf^\F|^2+\frac{D}{r^2} |T\dk\qf^\F|^2 \right)\\
 &&+\int_{\MM_{\ntrap}} r^{-3} \big(|e_4\dk\qf^\F|^2+ r^{-1} |\dk\qf^\F|^2\big) + r^{-1} |\nabb \dk\qf^\F|^2+r^{-1-\de}  |e_3(\dk\qf^\F)|^2 \\
 &&+\int_{\MM_{\ge  R}(\tau_1, \tau_2) }  r^{p-1}  \left( p | \ec_4(\dk\qf^\F) |^2 +(2-p)   |\nabb \dk\qf^\F|^2+   r^{-2}  |\dk\qf^\F|^2\right)
\eeaa
Consider  $\II_p[\qf, e C_1[\qf^\F]](\tau_1, \tau_2)$, where we can use the schematic version $C_1[\qf^\F]=\Big[\frac{1}{r}\lapp_2\qf^{\F}, \frac{1}{r^2}   e_4\qf^{\F},\frac{1}{r^2} e_3\qf^{\F}, \frac{1}{r^3}\qf^{\F}\Big]$. We consider the spacetime integral in the three region in which the spacetime is divided: $\MM_{red}$, $\MM_{trap}$ and $\MM_{far}$.

 In the redshift region $\MM_{red}$, the Morawetz bulk $\MM^1_p[\qf^\F]$ contains all second derivatives, and powers of $r$ don't matter. Therefore
\beaa
\II_p[\qf,  e C_1[\qf^\F]]&\simeq& e^2 \int_{\MM_{red}(\tau_1,\tau_2)}|\dkb^2\qf^{\F}|^2+ |\dk\qf^{\F}|^2+ |\qf^{\F}|^2 \les e^2 \MM^1_p[\qf^\F](\tau_1,\tau_2)
 \eeaa
 
  In the trapping region $\MM_{trap}$, we have
\beaa
\II_p[ \qf, e C_1[\qf^\F]]&\simeq& e \int _{\tau_1}^{\tau_2}  d\tau \int_{\Si_{trap}(\tau)  }  ( |R\qf|   +r^{-1} |\qf| ) (|\lapp_2\qf^{\F}|+| e_4\qf^{\F}|+|e_3\qf^{\F}|+|\qf^{\F}|)
\eeaa
All the terms with up to one derivative of $\qf^\F$ can be absorbed by
\beaa
e \int _{\tau_1}^{\tau_2}  d\tau \int_{\Si_{trap}(\tau)  }  ( |R\qf|  +r^{-1} |\qf| ) (| e_4\qf^{\F}|+|e_3\qf^{\F}|+|\qf^{\F}|)&\les& e \int_{\tau_1}^{\tau_2} E[\qf](\tau)^{1/2}\big( \int_{\Si_{trap}(\tau)} |\dk \qf^\F|^2\big)^{1/2}\\
&\les& e \sup_{\tau\in[\tau_1, \tau_2]} E[\qf](\tau)+e \MM^1_p[\qf^\F](\tau_1, \tau_2)
\eeaa
Similarly, since $R\qf$ does not degenerate in the trapping region, we can bound
\beaa
e \int _{\tau_1}^{\tau_2}  d\tau \int_{\Si_{trap}(\tau)  }  ( |R\qf|  +r^{-1} |\qf| ) |\lapp_2\qf^{\F}|&\les& e \int_{\tau_1}^{\tau_2} E^1[\qf^\F](\tau)^{1/2}\big( \int_{\Si_{trap}(\tau)}|\Rbrev\qf|^2  +r^{-1} |\qf|^2\big)^{1/2}\\
&\les& e \sup_{\tau\in[\tau_1, \tau_2]} E^1[\qf^\F](\tau)+e \MM_p[\qf](\tau_1, \tau_2)
\eeaa
 In the far-away region $\MM_{far}$, we write
\beaa
\II_p[\qf, e C_1[\qf^\F]](\tau_1,\tau_2)&=& e^2\int_{\MM_{far}(\tau_1,\tau_2)} r^{1+p}  |\frac{1}{r}\lapp_2\qf^{\F}, \frac{1}{r^2}   e_4\qf^{\F},\frac{1}{r^2} e_3\qf^{\F}, \frac{1}{r^3}\qf^{\F}|^2 \\
&=& e^2 \int_{\MM_{far}(\tau_1,\tau_2)} r^{-5+p}|\dkb^2\qf^{\F}|^2+r^{-3+p} |e_4\qf^{\F}|^2+r^{-3+p}| e_3\qf^{\F}|^2+r^{-5+p}|\qf^{\F}|^2
\eeaa
On the other hand in the far-away region, for $\de\le p \le 2-\de$, the Morawetz bulk simplifies to
\beaa
\MM^1_p[\qf^\F](\tau_1,\tau_2)&\simeq& \int_{\MM_{far}(\tau_1, \tau_2) }  r^{-3+p}|\dkb \dk\qf^\F|^2+r^{-1-\de}  |e_3\dk\qf^\F|^2+ r^{-1+p}   | e_4\dk\qf^\F |^2 \\
&&+\int_{\MM_{far}(\tau_1, \tau_2) }r^{-1+p}   | e_4\qf^\F |^2 + r^{-1-\de}  |e_3\qf^\F|^2+   r^{-1+p}|\nabb \qf^\F|^2+ r^{-3+p} |\qf^\F|^2
\eeaa
 Since the powers of $r$ of $\II_p[ \qf, e C_1[\qf^\F]]$ decay all faster than the respective ones in $\MM^1_p[\qf^\F]$, we have that, for $r\ge R_0$,
 \beaa
 \II_p[ \qf, e C_1[\qf^\F]](\tau_1, \tau_2) &\les&e^2 \MM^1_p[\qf^\F](\tau_1,\tau_2)
 \eeaa

Consider $\II^1_p[\qf^\F, C_2[\qf]]$ with $C_2[\qf]= \frac {1}{ r^3} \qf$. We separate this norm in the three regions.

 In the redshift region $\MM_{red}$, we have
\beaa
\II^1_p[\qf^\F, e C_2[\qf]](\tau_1, \tau_2)&\simeq& e^2 \int_{\MM_{red}(\tau_1,\tau_2)}  |\qf|^2  +e^2 \int_{\MM_{red}(\tau_1,\tau_2)} |\dk \qf|^2\les e^2\MM_p[\qf](\tau_1, \tau_2)
\eeaa
 In the trapping region $\MM_{trap}$, we have 
\beaa
\II^1_p[\qf^\F, e C_2[\qf]](\tau_1, \tau_2)&=&e \int _{\tau_1}^{\tau_2}  d\tau \int_{\Si_{trap}(\tau)  }  ( |R\dk\qf^\F| +r^{-1} |\dk\qf^\F| ) |\dk\qf|
\eeaa
The last term, with only one derivative of $\qf^\F$ can be easily bounded:
\beaa
e\int _{\tau_1}^{\tau_2}  d\tau \int_{\Si_{trap}(\tau)  }  |\dk\qf^\F| |\dk\qf|&\les& e\int_{\tau_1}^{\tau_2} E[\qf](\tau)^{1/2}\big( \int_{\Si_{trap}(\tau)} |\dk \qf^\F|^2\big)^{1/2}\\
&\les& e\sup_{\tau\in[\tau_1, \tau_2]} E[\qf](\tau)^{1/2}\int_{\tau_1}^{\tau_2}\big( \int_{\Si_{trap}(\tau)} |\dk \qf^\F|^2\big)^{1/2}\\
&\les& e\sup_{\tau\in[\tau_1, \tau_2]} E[\qf](\tau)+e\MM^1_p[\qf^\F](\tau_1, \tau_2)
\eeaa
Similarly, the term $|R(\dk\qf^\F)|^2$ is present in the Morawetz bulk of $\qf^\F$ without degeneracy, therefore as before 
\beaa
e\int _{\tau_1}^{\tau_2}  d\tau \int_{\Si_{trap}(\tau)  } |\Rbrev\dk\qf^\F| |\dk\qf|&\les& e\sup_{\tau\in[\tau_1, \tau_2]} E[\qf](\tau)+e\MM^1_p[\qf^\F](\tau_1, \tau_2)
\eeaa
 In the far-away region $\MM_{far}$, we need to take into account the power of $r$ in the structure of the coupling terms \eqref{schematic-coupling-terms}. We have 
\beaa
\II^1_p[\qf^\F, eC_2[\qf]]&=& e^2\int_{\MM_{far}(\tau_1,\tau_2)} r^{1+p}  |\frac {1}{ r^3} \qf|^2  + e^2\int_{\MM_{far}(\tau_1,\tau_2)} r^{1+p}  |\frac {1}{ r^3} \dk\qf|^2 \\
&=& e^2\int_{\MM_{far}(\tau_1,\tau_2)} r^{-5+p}  |\qf|^2 + r^{-5+p}  | e_3\qf|^2+  r^{-3+p}  | e_4\qf|^2+  r^{-3+p}  | \nabb\qf|^2
\eeaa
On the other hand in the far-away region, for $\de\le p \le 2-\de$, the Morawetz bulk simplifies to
\beaa
\MM_p[\qf](\tau_1,\tau_2)&=& \int_{\MM_{far}} r^{-3} \big(|e_4\qf|^2+ r^{-1} |\qf|^2\big) + r^{-1} |\nabb \qf|^2+r^{-1-\de}  |e_3(\qf)|^2 \\
 &&+\int_{\MM_{far}(\tau_1, \tau_2) }  r^{p-1}  \left( p | \ec_4(\qf) |^2 +(2-p)   |\nabb \qf|^2+   r^{-2}  |\qf|^2\right)\\
 &\simeq& \int_{\MM_{far}(\tau_1, \tau_2) }  r^{-3+p} |\qf|^2+r^{-1-\de}  |e_3\qf|^2+ r^{-1+p}   | e_4\qf |^2 +  r^{-1+p}|\nabb \qf|^2  
 \eeaa
 Since the powers of $r$ of $\II^1_p[ C_2[\qf]]$ decay all faster than the respective ones in $\MM_p[\qf]$, we have that, for $r\ge R_0$,
 \beaa
 \II^1_p[ \qf^\F, eC_2[\qf]](\tau_1, \tau_2) &\les&e^2 \MM_p[\qf](\tau_1,\tau_2)
 \eeaa
 
 Finally we bound the last term commuting $T$ and $R$:
 \beaa
 -e \Lambda \int_{\MM_{trap}}  T R \qf^\F \c R( C_2[\qf] )&\les& -e \Lambda \int_{\MM_{trap}}  (RT \qf^\F +\dk \qf^\F)\c R( \qf) \les e\sup_{\tau\in[\tau_1, \tau_2]} E[\qf](\tau)+e\MM^1_p[\qf^\F](\tau_1, \tau_2)
 \eeaa
as before. Combining all the above bounds we get the desired estimate. 
\end{proof}

Observe that the proof of Proposition \ref{absorption-coupling-terms-1} fails in the trapping region for the terms involving $T\qf \ \lapp \qf^\F$, $TT\qf^\F \ T\qf$ and $T\nabb\qf^\F \ \nabb\qf$, because these terms are degenerate in the Morawetz bulks. To resolve this problem, we will make use of the equations, and the particular structure of the coupling terms implies a cancellation of the degenerate terms.

\begin{proposition}\label{absorption-coupling-terms-2} With the notations above, for all $\de \le p \le 2-\de$, we have
\beaa
&&-e \Lambda \int_{\MM_{trap}}2  T( \qf )\c  C_1[\qf^\F]+ T T \qf^\F \c T( C_2[\qf])+  T \dkb \qf^\F \c \dkb( C_2[\qf]) \\
&\les&  e \sup_{\tau\in[\tau_1, \tau_2]} \left( E[\qf](\tau)+ E^1[\qf^\F](\tau)\right)+e\left(\MM_p[\qf]+\MM^1_p[\qf^\F] + \hat{\MM}[\psi_3]\right)(\tau_1, \tau_2)
\eeaa
\end{proposition}
\begin{proof} We write explicitely the terms $C_1[\qf^\F]$ and $C_2[\qf]$ and since $T(r)=O(\ep)$ and $\nabb(r)=0$, we have
\beaa
&&-e \Lambda \int_{\MM_{trap}}2  T( \qf )\c  C_1[\qf^\F]+ T T \qf^\F \c T( C_2[\qf])+  T \dkb \qf^\F \c \dkb( C_2[\qf])\\
&=&-e \Lambda \int_{\MM_{trap}}2  T( \qf )\c  \left(\frac{2}{r}\lapp_2\qf^{\F}-\frac{2}{r}\kab   e_4\qf^{\F}-\frac{2}{r}\ka e_3\qf^{\F} + \frac 1 r \left(3\ka\kab+8\rho+4\rhoF^2\right)\qf^{\F}\right)-\frac {2}{ r^3} T T \qf^\F \c T( \qf)-\frac {2}{ r^3}  T \dkb \qf^\F \c \dkb(\qf)
\eeaa
The terms involving only one derivative of $\qf^\F$ can be bounded by their absolute value and by the Morawetz bulk as before:
\beaa
-e \Lambda \int_{\MM_{trap}}2  T( \qf )\c  \left(-\frac{2}{r}\kab   e_4\qf^{\F}-\frac{2}{r}\ka e_3\qf^{\F} + \frac 1 r \left(3\ka\kab+8\rho+4\rhoF^2\right)\qf^{\F}\right)\les e\sup_{\tau\in[\tau_1, \tau_2]} E[\qf](\tau)+e\MM^1_p[\qf^\F](\tau_1, \tau_2)
\eeaa
The higher order terms that are degenerate in the bulks are instead
\beaa
 \int_{\MM_{trap}}-2  T( \qf )\c  \left(\frac{2}{r}\lapp_2\qf^{\F}\right)+\frac {2}{ r^3} T T \qf^\F \c T( \qf)+\frac {2}{ r}  T \nabb \qf^\F \c \nabb(\qf)
\eeaa
where $\dkb=r\nabb$. The last term can be integrated by parts twice estimating the boundary terms by the energy, obtaining 
\beaa
\int_{\MM_{trap}} \frac {2}{ r}  T \nabb \qf^\F \c \nabb(\qf)&\les&\int_{\MM_{trap}} -\frac {2}{ r}  \nabb \qf^\F \c T \nabb(\qf)+\sup_{\tau\in[\tau_1, \tau_2]} E[\qf](\tau)+\sup_{\tau\in[\tau_1, \tau_2]} E^1[\qf^\F]\\
&\les& \int_{\MM_{trap}} \frac {2}{ r}  \lapp_2 \qf^\F \c T(\qf)+\sup_{\tau\in[\tau_1, \tau_2]} E[\qf](\tau)+\sup_{\tau\in[\tau_1, \tau_2]} E^1[\qf^\F]
\eeaa
Finally we can rewrite the second term using the wave equation for $\qf^\F$. In fact we write $\square_2$ in terms of the vectorfields $T$ and $R$, ignoring the quadratic terms:
\beaa
 \square_2  \Psi  &=& -\frac 1 \Up TT\Psi+\frac 1 \Up RR\Psi +\lapp_2\Psi +\frac 2 r  R \Psi
 \eeaa
Renormalizing $T$ as $\tilde{T}=\frac{r}{\Up^{1/2}}T$, we will obtain for $\qf^\F$
\beaa
\tilde{T}\tilde{T}\qf^\F  &=&r^2\lapp_2\qf^\F +\frac {r^2}{ \Up} RR\qf^\F + 2 r  R \qf^\F -r^2 \square_2\qf^\F \\
 &=&r^2\lapp_2\qf^\F +\frac {r^2}{ \Up} RR\qf^\F + 2 r  R \qf^\F -r^2 V_2 \qf^\F+ e \frac{2}{r} \qf -e^2 \frac{4}{r}\psi_3
\eeaa
Without loss of generality, we can derive all the above estimate with the normalization $\tilde{T}$ in the trapping region. In particular, in Proposition \ref{1st-order-estimates}, we can commute the equation with a deformation of $T$ that coincides with $\tilde{T}$ in the trapping region. With an abuse of notation, we denote this deformation again $T$. Using the equation, and estimating the other terms by the Morawetz bulks we have
\beaa
\int_{\MM_{trap}}\frac {2}{ r^3} T T \qf^\F \c T( \qf)&=& \int_{\MM_{trap}}\frac {2}{ r^3} \left( r^2\lapp_2\qf^\F +\frac {r^2}{ \Up} RR\qf^\F + 2 r  R \qf^\F -r^2 V_2 \qf^\F+ e \frac{2}{r} \qf -e^2 \frac{4}{r}\psi_3\right) \c T( \qf)\\
&\les& \int_{\MM_{trap}}\frac {2}{ r} \lapp_2\qf^\F \c T( \qf)\\
&&+ \sup_{\tau\in[\tau_1, \tau_2]} E[\qf](\tau)+e\MM_p[\qf](\tau_1, \tau_2)+\MM^1_p[\qf^\F](\tau_1, \tau_2) +e^2 \hat{\MM}[\psi_3](\tau_1, \tau_2)
\eeaa
We observe the cancellation for the higher order terms $\int_{\MM_{trap}}-\frac{4}{r} T( \qf )\c  \lapp_2\qf^{\F}+\frac {2}{ r} \lapp_2\qf^\F \c T( \qf)+\frac {2}{ r} \lapp_2\qf^\F \c T( \qf)$, therefore
\beaa
&& e\Lambda \int_{\MM_{trap}}-2  T( \qf )\c  \left(\frac{2}{r}\lapp_2\qf^{\F}\right)+\frac {2}{ r^3} T T \qf^\F \c T( \qf)+\frac {2}{ r}  T \nabb \qf^\F \c \nabb(\qf) \\
&\les&e \sup_{\tau\in[\tau_1, \tau_2]} \left( E[\qf](\tau)+ E^1[\qf^\F](\tau)\right)+e\left(\MM_p[\qf]+\MM^1_p[\qf^\F] + \hat{\MM}[\psi_3]\right)(\tau_1, \tau_2)
\eeaa
which gives the desired estimate.
\end{proof}

\begin{remark}\label{diagonalizable} The above cancellation is related to the structure of the equations, in particular to the fact that the system formed by the higher order terms is diagonalizable, as observed by Pei-Ken Hung. In fact consider the higher order terms of the system \eqref{schematic-system}: 
\bea\label{higher-order-system}
\begin{cases}
\square_2\qf= e \frac{2}{r^3} \dkb^2 \qf^\F, \\
\square_2\qf^{\F}=-e \frac{2}{r^3} \qf
\end{cases}
\eea
Commuting the second equation with $\dkb$ and denoting $\frak{p}=\dkb \qf^\F$ we obtain the system, up to lower order terms,
\beaa
\begin{cases}
\square_2\qf= e \frac{2}{r^3} \dkb \frak{p}, \\
\square_2\frak{p}=-e \frac{2}{r^3} \dkb\qf
\end{cases}
\eeaa
Therefore constructing the stress-energy tensor as $\QQ_{ab}=\QQ[\qf]_{ab}+\QQ[\frak{p}]_{ab}+e\frac{2}{r^3} \dk \qf \c \frak{p} \ \g_{ab}$, in taking the divergence of $\QQ_{ab} X^b$, for any vector field $X$, upon integration we will get a cancellation of the higher order terms of the equation with the divergence of the added term $e\frac{2}{r^3} \dk \qf \c \frak{p} \ \g_{ab}$. Observe that this is the same structure observed in \cite{Pei-Ken} for the system governing the odd part of the gravitational perturbation of Schwarzschild spacetime using harmonic gauge. A very interesting feature of this structure is that it doesn't need smallness of the right hand side, because the cancellation holds at the level of the stress-energy tensor.
Neverthless the system \eqref{schematic-system} contains lower order terms which don't have the same diagonalizable structure. Because of the presence of the first derivatives on the right hand side of the first equation, we need to commute the second equation with all derivatives, and at the present day we weren't able to find a way to cancel those terms as in the simpler case of \eqref{higher-order-system}. We consider instead the stress-energy tensor of the separated equations and get the cancellation of the higher order terms after obtaining the estimates, as shown in Proposition \ref{absorption-coupling-terms-1}. In this approach though we need smallness of the charge to absorb the terms on the right hand side. We expect that with a more careful analysis of the lower order terms, this hypothesis could be eliminated.
\end{remark}

\subsection{Transport estimates for the lower order terms}\label{section-lower-order}
The aim of this section is to prove estimate \eqref{absorb-lot} in Step 3 of the proof of Theorem \ref{main-theorem-1}.

By \eqref{definition-main-coefficients}, the lower order terms of the equations are schematically given by
 \beaa
 L_1[\qf]&=& \Big[\frac{1}{r^3}\psi_1, \frac{1}{r^2} \psi_0\Big], \\
 L_1[\qf^\F]= L_2[\qf^\F]&=& \Big[\frac{1}{r^3} \psi_3 \Big]
 \eeaa
 Ignoring quadratic terms, the operator $\underline{P}$ can be written as $\underline{P}f=\kab^{-1} e_3(rf)$, therefore by \eqref{quantities}, we have 
\bea\label{differential-relations}
\bsplit
e_3(r\psi_0)=-\frac 2 r  \psi_1, \\
e_3(r\psi_1)=-\frac 2 r \qf, \\
e_3(r\psi_3)=-\frac 2 r \qf^\F
\end{split}
\eea
We derive transport estimates for $\psi_0$, $\psi_1$ and $\psi_3$ using the above differential relations, in the same way as done in \cite{TeukolskyDHR}.

\begin{proposition}\label{transport-estimates} Let $\psi_0$, $\psi_1$, $\qf$, $\psi_3$, $\qf^\F$ be defined as in \eqref{quantities}. Then, for all $\de \le p \le 2-\de$, we have
\beaa
(E_p[\psi_0]+E_p[\psi_1])(\tau_2) +(\hat{\MM}_p[\psi_1]+\hat{\MM}_p[\psi_0])(\tau_1, \tau_2)&\les&  (E_p[\psi_0]+E_p[\psi_1])(\tau_1)+ \sup_{\tau\in[\tau_1, \tau_2]} E[\qf](\tau)\\
E_p[\psi_3](\tau_2)+\hat{\MM}_p[\psi_3](\tau_1, \tau_2)&\les& E_p[\psi_3](\tau_1)+ \sup_{\tau\in[\tau_1, \tau_2]} E^1[\qf^\F](\tau)
\eeaa
\end{proposition}
\begin{proof} We first consider estimates for $\psi_3$. From $e_3(r\psi_3)=-\frac 2 r  \qf^\F$, we derive
\bea\label{div1}\begin{split}
\div(r^n |r\psi_3|^2 e_3)&= e_3(r^n |r\psi_3|^2)+r^n |r\psi_3|^2\div(e_3)\\
&=-n r^{n-1}|r\psi_3|^2-4r^{n-1}\qf^\F |r\psi_3|+\frac 1 2 r^n |r\psi_3|^2 \tr \pi^{(e_3)}\\
&=-(n+2) r^{n-1}|r\psi_3|^2-4r^{n-1}\qf^\F |r\psi_3| 
\end{split}
\eea
Similarly, commuting \eqref{differential-relations} with $re_4$ and $r \nabb$ we get 
\bea\label{div2}
\begin{split}
\div(r^n |re_4(r\psi_3)|^2 e_3)&=  -(n+2) r^{n-1} |re_4(r\psi_3)|^2+2r^n|re_4(r\psi_3)|re_4 (- \frac 2r   \qf^\F)+2r^n |re_4(r\psi_3)|[e_3,re_4](r\psi_3)\\
&=  -(n+4) r^{n-1}|r e_4(r\psi_3)|^2+r^n|r e_4(r\psi_3)|\left(\frac 4 r -\frac{16\varpi}{r^2}+\frac{12 e^2}{r^3}  \right) \qf^\F-4r^n|r e_4(r\psi_3)| e_4(\qf^\F), \\
\div(r^n |r\nabb(r\psi_3)|^2 e_3)&=  -(n+2) r^{n-1}|r \nabb(r\psi_3)|^2-4r^{n}|r \nabb(r\psi_3)|  \nabb(\qf^\F)
\end{split}
\eea
We add \eqref{div1} to \eqref{div2} and integrate in $\MM(\tau_1, \tau_2)$, with $n=p-4$. When integrating, in the trapping we bound:
\beaa
\int_{M_{trap}}|\qf^\F+\dk \qf^\F||\dk \psi_3|&\les& \int_{\tau_1}^{\tau_2} E[\qf](\tau)^{1/2}\big( \int_{\Si_{trap}(\tau)} |\dk \psi_3|^2\big)^{1/2}\\
&\les& \sup_{\tau\in[\tau_1, \tau_2]} E[\qf^\F](\tau)^{1/2}\int_{\tau_1}^{\tau_2}\big( \int_{\Si_{trap}(\tau)} |\dk \psi_3|^2\big)^{1/2}\\
&\les& \lambda\sup_{\tau\in[\tau_1, \tau_2]} E^1[\qf^\F](\tau)+\lambda^{-1}\hat{\MM}_p[\psi_3](\tau_1, \tau_2)
\eeaa
and for $\lambda$ big enough, the non-degenerate bulk $\hat{\MM}[\psi_3]$ can be absorbed by the left hand side of \eqref{div2}. We arrive therefore to 
\beaa
E_p[\psi_3](\tau_2)+\hat{\MM}_p[\psi_3](\tau_1, \tau_2)&\les& E_p[\psi_3](\tau_1)+ \sup_{\tau\in[\tau_1, \tau_2]} E^1[\qf^\F](\tau)
\eeaa

Similarly, using the differential relation $e_3(r\psi_1)=-\frac 2 r \qf$, we derive transport estimates for $\psi_1$:
\bea\label{transport-estimate-psi1-qf}
E_p[\psi_1](\tau_2) +\hat{\MM}_p[\psi_1](\tau_1, \tau_2)&\les & E_p[\psi_1](\tau_1)+ \sup_{\tau\in[\tau_1, \tau_2]} E[\qf](\tau)
\eea
Finally, using the differential relation $e_3(r\psi_0)=-\frac 2 r  \psi_1$, and bounding with Cauchy-Schwarz we obtain
\bea\label{transport-estimate-psi0-psi1}
E_p[\psi_0](\tau_2) +\hat{\MM}[\psi_0](\tau_1, \tau_2)&\les & E_p[\psi_0](\tau_1)+ \int_{\MM(\tau_1, \tau_2)} r^{-3+p}(|\psi_1 |^2+|\dk\psi_1|^2)
\eea
Multiplying \eqref{transport-estimate-psi0-psi1} by a constant $A$ and summing to \eqref{transport-estimate-psi1-qf}, choosing $A\ll 1$, we can absorb the integral on the right hand side of \eqref{transport-estimate-psi0-psi1} by the non-degenerate bulk of $\psi_1$. We obtain
\beaa
(E_p[\psi_0]+E_p[\psi_1])(\tau_2) +(\hat{\MM}_p[\psi_1]+\hat{\MM}_p[\psi_0])(\tau_1, \tau_2)&\les & (E_p[\psi_0]+E_p[\psi_1])(\tau_1)+ \sup_{\tau\in[\tau_1, \tau_2]} E[\qf](\tau)
\eeaa
as desired.
\end{proof}

We derive now the estimates for the terms of the main estimates involving the lower order terms.
\begin{proposition}\label{estimates-lot} With the notations above, for all $\de \le p \le 2-\de$, we have
\beaa
&&\left(\II_p[\qf, e   L_1[\qf^\F]]+ \II_p[\qf, e  L_1[\qf]]+\II^1_p[\qf^\F, e^2  L_2[\qf^\F]]\right)(\tau_1, \tau_2)  \\
                 && -e \Lambda \int_{\MM_{trap}}2 T( \qf )\c \left(L_1[\qf^\F]+e L_1[\qf] \right) - e T T \qf^\F \c T( L_2[\qf^\F] )-  e T R \qf^\F \c R( L_2[\qf^\F])- e  T \nabb \qf^\F \c \nabb( L_2[\qf^\F] )\\
&\les&e\sup_{\tau\in[\tau_1, \tau_2]} \left(E[\qf](\tau)+E^1[\qf^\F](\tau)\right)+e \left(\hat{\MM}[\psi_0]+\hat{\MM}[\psi_1]+\hat{\MM}[\psi_3]+\MM^1_p[\qf^\F]\right)(\tau_1, \tau_2)
\eeaa
\end{proposition}
\begin{proof} Consider $\II_p[\qf, e L_1[\qf^\F]]$, and recall that $L_1[\qf^\F]=\frac{1}{r^3}\psi_3$. We consider the spacetime integral in the three region in which the spacetime is divided: $\MM_{red}$, $\MM_{trap}$ and $\MM_{far}$.

 In the redshift region $\MM_{red}$, 
\beaa
\II_p[\qf, e L_1[\qf^\F]]&\simeq& e^2 \int_{\MM_{red}(\tau_1,\tau_2)}|\psi_3|^2\les e^2 \hat{\MM}_p[\psi_3](\tau_1, \tau_2) 
 \eeaa

 In the trapping region $\MM_{trap}$, we don't have issue with degeneracy, since the bulk term $\hat{M}[\psi_3]$ is not degenerate at the trapping region:
\beaa
\II_p[\qf, e L_1[\qf^\F]]&\simeq& e \int _{\tau_1}^{\tau_2}  d\tau \int_{\Si_{trap}(\tau)  }  ( |R\qf|  +r^{-1} |\qf| ) (|\psi_3|)\les e \int_{\tau_1}^{\tau_2} E[\qf](\tau)^{1/2}\big( \int_{\Si_{trap}(\tau)} |\psi_3|^2\big)^{1/2}\\
&\les& e\sup_{\tau\in[\tau_1, \tau_2]} E[\qf](\tau)+e\hat{\MM}_p[\psi_3](\tau_1, \tau_2) 
\eeaa
 In the far-away region $\MM_{far}$, 
\beaa
\II_p[\qf, eL_1[\qf^\F]](\tau_1,\tau_2)&=& e^2\int_{\MM_{far}(\tau_1,\tau_2)} r^{1+p}  |\frac{1}{r^3}\psi_3|^2 = e^2\int_{\MM_{far}(\tau_1,\tau_2)} r^{-5+p}|\psi_3|^2\\
&\les& e^2\hat{\MM}_p[\psi_3](\tau_1, \tau_2) 
\eeaa
which is absorbed by the non-degenerate bulk, which decays as $r^{-1+p}|\psi_3|^2$.

Consider $\II_p[\qf, e^2 L_1[\qf]]$ and recall that $L_1[\qf]= \Big[\frac{1}{r^3}\psi_1, \frac{1}{r^2} \psi_0\Big]$.

In the redshift region $\MM_{red}$, 
\beaa
\II_p[\qf, e^2 L_1[\qf]]&\simeq& e^4 \int_{\MM_{red}(\tau_1,\tau_2)}|\psi_1|^2+|\psi_0|^2 \les e^4\hat{\MM}_p[\psi_1](\tau_1, \tau_2) +e^4 \hat{\MM}_p[\psi_0](\tau_1, \tau_2)
 \eeaa
 In the trapping region $\MM_{trap}$, we have
\beaa
\II_p[\qf, e^2 L_1[\qf]]&\simeq& e^2 \int _{\tau_1}^{\tau_2}  d\tau \int_{\Si_{trap}(\tau)  }  ( |R\qf|    +r^{-1} |\qf| ) (|\psi_1|+|\psi_0|)\\
&\les& e^2 \int_{\tau_1}^{\tau_2} E[\qf](\tau)^{1/2}\big( \int_{\Si_{trap}(\tau)} |\psi_1|^2+|\psi_0|^2\big)^{1/2}\\
&\les& e^2\sup_{\tau\in[\tau_1, \tau_2]} E[\qf](\tau)+e^2\hat{\MM}_p[\psi_1](\tau_1, \tau_2) +e^2\hat{\MM}_p[\psi_0](\tau_1, \tau_2)
\eeaa
In the far-away region $\MM_{far}$, 
\beaa
\II_p[\qf, e^2 L_1[\qf]](\tau_1,\tau_2)&=& e^4\int_{\MM_{far}(\tau_1,\tau_2)} r^{1+p}  |\frac{1}{r^3}\psi_1, \frac{1}{r^2} \psi_0|^2 \\
&=& e^4\int_{\MM_{far}(\tau_1,\tau_2)} r^{-5+p}|\psi_1|^2+r^{-3+p} |\psi_0|\\
&\les& e^4\hat{\MM}_p[\psi_1](\tau_1, \tau_2) +e^4\hat{\MM}_p[\psi_0](\tau_1, \tau_2)
\eeaa
which is absorbed by the non-degenerate bulks, which decay as $r^{-1+p}(|\psi_1|^2+|\psi_0|^2)$.

Consider $\II^1_p[\qf^\F, e^2 L_2[\qf^\F]]$, and recall that $L_2[\qf^\F]=\frac{1}{r^3}\psi_3$. We have in the redshift region $\MM_{red}$, 
\beaa
\II^1_p[\qf^\F, e^2L_2[\qf^\F]]&\simeq& e^4\int_{\MM_{red}(\tau_1,\tau_2)}|\dk\psi_3|^2\les e^4\hat{\MM}_p[\psi_3](\tau_1, \tau_2) +e^4 \MM^1_p[\qf^\F](\tau_1, \tau_2)
 \eeaa
In the trapping region $\MM_{trap}$, we have
\beaa
\II^1_p[\qf^\F, e^2 L_2[\qf^\F]]&\simeq&e^2\int _{\tau_1}^{\tau_2}  d\tau \int_{\Si_{trap}(\tau)  }  ( |R\dk\qf^\F|  +r^{-1} |\dk\qf^\F| ) (|\dk\psi_3|)\\
&\les& e^2\int_{\tau_1}^{\tau_2} E^1[\qf^\F](\tau)^{1/2}\big( \int_{\Si_{trap}(\tau)} |\dk\psi_3|^2\big)^{1/2}\\
&\les& e^2\sup_{\tau\in[\tau_1, \tau_2]} E^1[\qf^\F](\tau)+e^2 \hat{\MM}_p[\psi_3](\tau_1, \tau_2) +e^4 \MM^1_p[\qf^\F](\tau_1, \tau_2)
\eeaa
In the far-away region $\MM_{far}$, 
\beaa
\II^1_p[\qf^\F, e^2 L_2[\qf^\F]](\tau_1,\tau_2)&=&  e^4\int_{\MM_{far}(\tau_1,\tau_2)} r^{-5+p}|\dk\psi_3|^2 \les e^4\hat{\MM}_p[\psi_3](\tau_1, \tau_2) +e^4 \MM^1_p[\qf^\F](\tau_1, \tau_2)
\eeaa
For the remaining terms we have similarly
\beaa
-e \Lambda \int_{\MM_{trap}}2 T( \qf )\c \left(L_1[\qf^\F]+e L_1[\qf] \right) &\les& e\sup_{\tau\in[\tau_1, \tau_2]} E[\qf](\tau)+e^2\hat{\MM}_p[\psi_1](\tau_1, \tau_2) \\
&&+e^2\hat{\MM}_p[\psi_0](\tau_1, \tau_2)+e\hat{\MM}_p[\psi_3](\tau_1, \tau_2), 
\eeaa
\beaa
&&e^2 \Lambda \int_{\MM_{trap}}  T T \qf^\F \c T( L_2[\qf^\F] )+ T R \qf^\F \c R( L_2[\qf^\F])+  T \nabb \qf^\F \c \nabb( L_2[\qf^\F] )\\
&\les& e^2\sup_{\tau\in[\tau_1, \tau_2]} E^1[\qf^\F](\tau)+e^2\hat{\MM}_p[\psi_3](\tau_1, \tau_2)
\eeaa
and combining all the above, we obtain the desired estimate.
\end{proof}

\begin{corollary}\label{corollary-lot} With the notations above, for all $\de \le p \le 2-\de$, we have
\beaa
&&\left(\II_p[\qf, e   L_1[\qf^\F]]+ \II_p[\qf, e  L_1[\qf]]+\II^1_p[\qf^\F, e^2  L_2[\qf^\F]]\right)(\tau_1, \tau_2)  \\
                 && -e \Lambda \int_{\MM_{trap}}2 T( \qf )\c \left(L_1[\qf^\F]+e L_1[\qf] \right) - e T T \qf^\F \c T( L_2[\qf^\F] )-  e T R \qf^\F \c R( L_2[\qf^\F])- e  T \nabb \qf^\F \c \nabb( L_2[\qf^\F] )\\
&\les&e\sup_{\tau\in[\tau_1, \tau_2]} \left(E[\qf](\tau)+E^1[\qf^\F](\tau)\right)+e \MM^1_p[\qf^\F](\tau_1, \tau_2)+ \left(E_p[\psi_0]+E_p[\psi_1]+E_p[\psi_3]\right)(\tau_1)
\eeaa
\end{corollary}
\begin{proof} Straightforward consequence of Proposition \ref{estimates-lot} and the transport estimates in Proposition \ref{transport-estimates}.
\end{proof}

\begin{flushleft}
\small{DEPARTMENT OF MATHEMATICS, COLUMBIA UNIVERSITY} \\
\textit{E-mail address}: \href{mailto:egiorgi@math.columbia.edu}{egiorgi@math.columbia.edu}
\end{flushleft}

\end{document}